%% file: ArXiV-15-2-2023.tex
\documentclass[11pt]{article}
\usepackage[utf8]{inputenc}
\usepackage{amsmath}
\usepackage{amsthm}
\usepackage{amssymb}
\usepackage[margin = 1in]{geometry}
\usepackage{todonotes}
\usepackage{multicol}
\usepackage[ruled,noend]{algorithm2e}
\usepackage{algorithmic} 
\usepackage[shortlabels]{enumitem}
\usepackage{comment}
\usepackage{tabularx}
\usepackage{makecell}
\usepackage{hyperref}

\usepackage{authblk}

\def\DEBUG{true}

\ifdefined\DEBUG

\else

\fi

\newtheorem{invariant}{Invariant}
\newtheorem{theorem}{Theorem}
\newtheorem{question}{Question}
\newtheorem{lemma}[theorem]{Lemma}
\newtheorem{corollary}[theorem]{Corollary}
\newtheorem{observation}[theorem]{Observation}
\newcommand{\eps}{\varepsilon}
\newcommand{\mx}{\max}
\newcommand{\mout}{\Delta(\overrightarrow{G})}
\newcommand{\polylog}{\operatorname{polylog}}
\newcommand{\poly}{\operatorname{poly}}
\title{Adaptive Out-Orientations with Applications}
\title{Adaptive Out-Orientations with Applications}
\author[1]{Aleksander B. G. Christiansen} \author[2]{Jacob Holm} \author[1]{Ivor van der Hoog} \author[1]{Eva Rotenberg} \author[3]{Chris Schwiegelshohn}
\affil[1]{Technical University of Denmark\thanks{Partially supported by Independent Research Fund Denmark grant 2020-2023 (9131-00044B) ``Dynamic Network Analysis'', the VILLUM Foundation grant 37507 ``Efficient Recomputations for Changeful Problems'' and the EUROTECH cofund (110531) .}}
\affil[2]{University of Copenhagen, Denmark\thanks{Partially supported by the VILLUM Foundation grant 16582, ``BARC".}}
\affil[3]{Aarhus University, Denmark\thanks{Supported by a Independent Research Fund Denmark (DFF) Sapere Aude Research
Leader grant No 1051-00106B}}
\date{}
\date{}

\begin{document}
\thispagestyle{empty}
\maketitle

\begin{abstract}
We give improved algorithms for maintaining edge-orientations of a fully-dynamic graph, such that the out-degree of each vertex is bounded. On one hand, we show how to orient the edges such that the out-degree of each vertex is proportional to the arboricity $\alpha$ of the graph, in a worst-case update time of $O(\log^3 n \log \alpha)$. On the other hand, motivated by applications including dynamic maximal matching, we obtain a different trade-off, namely the improved worst case update time of $O(\log ^2 n \log \alpha)$ for the problem of maintaining an edge-orientation with at most $O(\alpha + \log n)$ out-edges per vertex.
Since our algorithms have update times with worst-case guarantees, the number of changes to the solution (i.e. the recourse) is naturally limited.
Our algorithms 
adapt to the current arboricity of the graph, and yield improvements over previous work:

Firstly, we obtain an $O(\varepsilon^{-6}\log^3 n \log \rho)$ worst-case update time algorithm for maintaining a $(1+\varepsilon)$ approximation of the maximum subgraph density, $\rho$, improving upon the $O(\varepsilon^{-6}\log ^4 n)$ algorithm by Sawlani and Wang from STOC 2020.

Secondly, we obtain an $O(\varepsilon^{-6}\log^3 n \log \alpha)$  worst-case update time algorithm for maintaining a $(1 + \varepsilon)\textnormal{OPT} + 2$ approximation of the optimal out-orientation of a graph with adaptive arboricity $\alpha$, improving the $O(\varepsilon^{-6}\alpha^2 \log^3 n)$ algorithm by Christiansen and Rotenberg from ICALP 2022. 
This yields the first worst-case polylogarithmic dynamic algorithm for decomposing into $O(\alpha)$ forests.
Thirdly, we obtain arboricity-adaptive fully-dynamic deterministic algorithms for a varierty, of problems including maximal matching, $\Delta+1$ coloring, and matrix vector multiplication. All update times are worst-case $O(\alpha+\log^2n \log \alpha)$, where $\alpha$ is the current arboricity of the graph. 
%\textbf{These results lead to various improvements for algorithmic problems:}
Specifically for the maximal matching problem, this improves for $\alpha \in \Omega( \log n\sqrt{\log\log n})$, on the deterministic algorithms by Kopelowitz, Krauthgamer, Porat, and Solomon from ICALP 2014 running in time $O(\alpha^2 + \log^2 n)$ and by Neiman and Solomon from STOC 2013 running in time $O(\sqrt{m})$.
\end{abstract}
\setcounter{page}{0}

\thispagestyle{empty}

\newpage %STOC CALL 
\section{Introduction}
\label{sec:intro}
In dynamic graphs, one wishes to update a data structure over a graph $G(V,E)$ (or an answer to a specified graph problem) as the graph undergoes local updates such as edge insertions and deletions. 
One of the fundamental problems is to maintain an orientation of the edges such that the maximal out-degree over all nodes is minimized.
While the problem is interesting in its own right, bounded out-degree orientations have a number of applications. First, the problem is closely related to the task of finding the densest subgraph; indeed if the edges can be fractionally oriented, the maximum out-degree is equal to the density $\rho := \frac{|E\cap (S\times S)|}{|S|}$ of the densest subgraph $S\subseteq V$.
Secondly, bounded out-degree orientations appear frequently as subroutines for other problems. In particular, there exist a large body of work parameterizing the update time of dynamic algorithms for many fundamental problems such as shortest paths \cite{FrigioniMN03,KowalikK06}, vertex cover \cite{PelegS16,Solomon18}, graph coloring \cite{henzinger2020explicit,SolomonW20}, independent set \cite{OnakSSW20}, and, most prominently, maximum matching \cite{BernsteinS15,BernsteinS16,GSSU22,HeTZ14,NeimanS16,PelegS16,Solomon18} in terms of the arboricity $\alpha:=\max_{S\subseteq V,|S|\geq2} \lceil\frac{|E\cap (S\times S)|}{|S|-1}\rceil$.

In light of their widespread applicability, maintaining an orientation of the edge minimizing the maximal outdegree is extremely well motivated. In particular, we are interested in algorithm with worst-case deterministic update times, as these can be immediately used as black-box subroutines. In a recent breakthrough result, \cite{sawlani2020near} showed that it is possible to maintain an estimate for the smallest maximal outdegree in $\polylog(n)$ worst case deterministic time by maintaining an estimate for the density of the densest subgraph. Nevertheless, all known results for maintaining an orientation require at least update time $\Omega(\rho) = \Omega(\alpha)$ worst case update time, regardless of whether the algorithm is randomized or not \cite{KopelowitzKPS13}. For dense graphs, this bound may be arbitrarily close to $n$. Thus, it raises the following question:

\begin{question}
    Is it possible to maintain an (approximate) out-degree orientation in sublinear deterministic worst case update time?
\end{question}

\subsection{Our Contribution}

In this paper, we answer the aforementioned question in the affirmative. Specifically, we provide a framework for maintaining approximate out-orientations with various trade-offs between the quality of the out-degree orientation and update time. All of the algorithms are deterministic and have a $\polylog(n)$ worst case update time.
For the problem of maintaining an out-orientation we obtain:
\begin{enumerate}
    \item An orientation with maximal out-degree $O(\alpha)$ in update time $O(\log^3 n \log \alpha)$.
%    \item An orientation with maximal out-degree $(2+\varepsilon)\alpha +1$ in update time $O(\varepsilon^{-6}\log^3 n \log \alpha)$.
    \item An orientation with maximal out-degree $(1+\varepsilon)\alpha +2$ in update time $O(\varepsilon^{-6}\log^3 n \log \alpha)$.
    \item An orientation with maximal out-degree $O(\alpha +\log n)$ in update time $O(\log ^2 n \log \alpha)$.
\end{enumerate}

\noindent All times are worst-case. %HERE %yes %good.
The recourse, i.e. the number of re-orientations of edges, is in all three cases a $\log n$ factor lower than the running time.
We apply our framework to various other problems (see Table \ref{tab:results}.
%We briefly discuss some of the specific applications in the following.

\subparagraph{Densest Subgraph}
Using the duality between out-degree orientations and maximum subgraph density, we obtain a $(1+\varepsilon)$ approximate estimate for maximum subgraph density $\rho$ in update time $O(\varepsilon^{-6}\log^3 n \log \rho)$. 
This recovers (and in fact moderately improves) the recent algorithm by \cite{sawlani2020near} that has an update time of $O(\varepsilon^{-6}\log^4 n)$.
%Further, using the same update time, we can query in time $O(1)$ for any node whether it is included in the (approximate) densest subgraph, and output the (approximate) densest subgraph in $O(c)$ time where $c$ is the size of the output. T

\subparagraph{Arboricity Decomposition}
An arboricity decomposition partitions the edge set into a minimal number of forests. The best known dynamic algorithms for maintaining an $O(\alpha)$ 
arboricity decomposition has an amortized deterministic update time of $O(\log^2 n)$ due to \cite{henzinger2020explicit} and a $O(\sqrt{m}\log n)$ worst case deterministic update time due to \cite{Banerjee}. 
Distinguishing between arboricities $1$ and $2$ requires $\Omega(\log n)$ time \cite{Patrascu10,Banerjee}.
Using our framework, we substantially improve the worst case update time to $O(\log^3 n \log \alpha)$.

\subparagraph{Dynamic Matrix Vector Multiplication}
In the Dynamic Matrix Vector Multiplication problem, we are given an $n\times n$ matrix $A$ and an $n$-vector $x$. Our goal is to quickly maintain $y=Ax$ in the sense that we can quickly query every entry $y_i = (Ax)_i$, subject to additive updates to the entries of $x$ and $A$. 
%The problem is related to the Online Boolean Matrix-Vector Multiplication (OMV), which is commonly used to obtain conditional lower bounds \cite{CliffordGL15,HenzingerKNS15,LarsenW17,Patrascu10}.
Interpreting $A$ as the adjacency matrix of a graph, \cite{KopelowitzKPS13} presented an algorithm that supports updates to $A$ in time $O(\alpha^2+\log^2 n)$ and updates to $x$ in time $O(\alpha + \log n)$. We improve the update times to $A$ to $O(\alpha + \log ^2 n \log \alpha)$, which is better when $\alpha\in\Omega(\log n\sqrt{\log\log n})$.
%\chr{I think I need to understand this result and the implication of it better. Is it true that we could use an algorithm for this problem to solve OMv or Partrascu's multiphase problem? If so, I think we should mention it.}

\subparagraph{Maximal Matching}
A matching is a set of node disjoint edges. A matching $M$ is maximal if no edge of the graph can be added to it without violating the matching property. 
%The task of maintaining a maximal matching under dynamic updates has received substantial attention, as it provides a alternative to the infeasible problem of maintaining an optimal matching, and it has applications to other problems. Most notably, the ability to maintain a maximal matching immediately implies a the ability to maintain a $2$-approximate vertex cover.
More so than perhaps any other problem, there exists a large gap between the performance of the state of the art deterministic algorithms vs the state of the art randomized algorithms. Using randomization, one can achieve a $O(1)$ amortized \cite{Solomon16} and a $\polylog(n)$ worst case update time \cite{BernsteinFH21}. Deterministic algorithms so far can only achieve a $O(\sqrt{m})$ update time for arbitrary graphs \cite{NeimanS16}, a $O(\alpha^2 + \log^2 n)$ update time where $O(\alpha)$ is the current arboricity of the graph \cite{KopelowitzKPS13}. 
%and a $O(\alpha_{\mx}+ \log n)$ update time \cite{berglinetal:LIPIcs:2017:8263}, assuming that arboricity of the graph is always upper bounded by $\alpha_{\mx}$.
Using our framework, we improve on known deterministic algorithms whenever $\alpha \in \Omega(\log n \cdot \sqrt{\log\log n})$ by achieving an update time of $O(\alpha + \log^2 n\log \alpha)$.

\subparagraph{$\Delta+1$ Coloring}
A fundamental question in many models of computation is how to efficiently compute a $\Delta+1$ coloring where $\Delta$ is the maximum degree of the graph. 
We present a deterministic algorithm that maintains a $\Delta+1$ coloring in $O(\alpha+\log^2 n\log \alpha)$ worst case update time. To the best of our knowledge, this is the first such algorithm that beats the trivial $O(\Delta)$ update time for uniformly sparse graphs. All other results~\cite{BhattacharyaGKL22,HenzingerP22,BhattacharyaCHN18,SolomonW20} require randomization, amortization, or do not yield a $\Delta+1$ coloring.

\subsection{Related Work}

\paragraph{Out-orientations.}
Historically, four criteria are considered when designing dynamic out-orientation algorithms: the maximum out-degree, the update time (or the recourse), amortized versus worst-case updates, and the adaptability of the algorithm to the current arboricity. 

Brodal and Fagerberg~\cite{Brodal99dynamicrepresentations} were the first to consider the out-orientation problem in a dynamic setting. 
They showed how to maintain an $\mathcal{O}(\alpha_{\max})$ out-orientation with an amortized update time of $\mathcal{O}(\alpha_{\max}+ \log{n})$, where $\alpha_{\max}$ is the maximum arboricity throughout the entire update sequence.
Thus, their result is adaptive to the current arboricity as long as it only increases. 
He, Tang, and Zeh~\cite{HeTZ14} and Kowalik~\cite{10.1016/j.ipl.2006.12.006} provided different analyses of Brodal and Fagerbergs algorithm resulting in faster update times at the cost of worse bounds on the maximum out-degree of the orientations. 
Henzinger, Neumann, and Wiese~\cite{henzinger2020explicit} gave an algorithm able to adapt to the current arboricity of the graph, achieving an out-degree of $\mathcal{O}(\alpha)$ and an amortized update time \emph{independent} of $\alpha$, namely $O(\log^2 n)$.
Kopelowitz, Krauthgamer, Porat, and Solomon~\cite{KopelowitzKPS13} showed how to maintain an $\mathcal{O}(\alpha+\log n)$ out-orientation with a worst-case update time of $\mathcal{O}(\alpha^2 + \log^2 n)$ fully adaptive to the arboricity. 
Christiansen and Rotenberg~\cite{christiansenICALP,christiansenMFCS} lowered the maximum out-degree to $(1+\varepsilon)\alpha+2$ incurring a worse update time of $\mathcal{O}(\varepsilon^{-6}\alpha^2\log^3 n)$.
Finally, Brodal and Berglin~\cite{berglinetal:LIPIcs:2017:8263} gave an algorithm with a different trade-off; they show how to maintain an $\mathcal{O}(\alpha_{\max}+\log n)$ out-orientation with a worst-case update time of $\mathcal{O}(\log n)$. This update time is faster and independent of $\alpha$, however the maximum out-degree does not adapt to the current value of $\alpha$.
Work on densest subgraph \cite{BahmaniKV12,BhattacharyaHNT15,EpastoLS15} can be used to estimate maximum degree of the best possible out orientation.
Sawlani and Wang~\cite{sawlani2020near} maintain a $(1 - \eps)$-approximate densest subgraph in worst-case time $O(\eps^{-6}\log^4 n )$ per update where they maintain an \emph{implicit} representation of the approximately-densest subgraph. They write that they can, in $O(\log n)$ time, identify the subset $S \subseteq V$ where $G[S]$ is the approximately-densest subgraph and they can report it in $O(|S|)$ additional time. Through a weak duality between densest subgraphs and  out-orientations, they can implicitly maintain a $(2 + \eps)$-approximate minimal fractional out-orientation in $O(\log^4 n \eps^{-6})$ per update. 

The out-orientation results by Brodal and Fagerberg~\cite{Brodal99dynamicrepresentations} and
Kopelowitz, Krauthgamer, Porat, and Solomon~\cite{KopelowitzKPS13} entail, through a well known reduction to the maximal matching problem (compare Theorem~\ref{thm:folklore} in the appendix), a dynamic maximal matching algorithm.
As mentioned, Theorem~\ref{thm:folklore} does not apply to the amortized update-time algorithms. 
In addition, the result by Sawlani and Wang~\cite{sawlani2020near} cannot be applied to maximal matchings. 
This is because their result does not maintain an explicit out-orientation. 
Rather, they guess $O(\log n)$ values $\alpha_{guess}$ for $\alpha$, and maintain an implicit out-orientation for each  $\alpha_{guess}$.
As the arboricity changes, one would have to change between versions of $\alpha_{guess}$, making it hard to bound the recourse $r_u$ needed in Theorem~\ref{thm:folklore}.
It is not possible to efficiently maintain a maximal matching for each of the $O(\log n)$ candidate out-orientations.

%s\clearpage
\paragraph{Maximal Matching}
Matchings have been widely studied in dynamic graph models. Under various plausible conjectures, we know that a maximum matching cannot be maintained even in the incremental setting and even for low arboricity graphs (such as planar graphs) substantially faster than $\Omega(n)$ update time \cite{AbboudD16,AbboudW14,HenzingerKNS15,KopelowitzPP16,Dahlgaard16}.
Given this, we typically relax the requirement from maximum matching to maintaining matchings with other interesting properties. 
Such a relaxation is to require that the maintained matching is only \emph{maximal}. The ability to retain a maximal matching is frequently used by other algorithms, notably it immediately implies a $2$-approximate vertex cover. 
In incremental graphs, maintaining a maximal matching is trivially done with the aforementioned greedy algorithm. 
For decremental\footnote{Maintaining an \emph{approximate maximum matching} decrementally is substantially easier than doing so for fully dynamic graphs. Indeed, recently work by \cite{AssadiBD22} matches the running times for approximate maximum matching in incremental graphs \cite{GLSSS19}. However, for maximal matching, we are unaware of work on decremental graphs that improves over fully dynamic results.} or fully dynamic graphs, there exist a number of trade-offs (depending on whether the algorithm is randomized or determinstic, and whether the update time is worst case or amortized).  Baswana, Gupta, and Sen~\cite{BaswanaGS15} and
Solomon~\cite{Solomon16} gave randomized algorithms maintaining a maximal matching with $O(\log n)$ and $O(1)$ amortized update time. These results were subsequently deamortized by Bernstein, Forster, and Henzinger \cite{BernsteinFH21} with only a $\polylog n$ increase in the update time. For deterministic algorithms, maintaining a maximal matching is substantially more difficult.
Ivkovic and Lloyd~\cite{IvkovicL93} gave a deterministic algorithm with $O((n+m)^{\sqrt{2}/2})$ worst case update time. This was subsequently improved to $O(\sqrt{m})$ worst case update time by Neiman and Solomon \cite{NeimanS16}, which remains the fastest deterministic algorithm for general graphs.

Nevertheless, there exist a number of results improving this result for low-arboricity graphs. Neiman and Solomon \cite{NeimanS16} gave a deterministic algorithm that, assuming that the arboriticty of the graph is always bounded by $\alpha_{\max}$, maintains a maximal matching in amortized time $O(\min_{\beta>1}\{\alpha_{\max} \cdot \beta + \log_{\beta} n\})$, which can be improved to $O(\log n/\log\log n)$ if the arboricity is always upper bounded by a constant. Under the same assumptions, He, Tang, and Zeh \cite{HeTZ14} improved this to $O(\alpha_{\max} + \sqrt{\alpha_{\max}\log n})$ amortized update time.
Without requiring that the arboricity be bounded at all times, the work by Kopelowitz, Krauthgsamer, Porat, and Solomon~\cite{KopelowitzKPS13} implies a deterministic algorithm with $O(\alpha^2 + \log^2 n)$ worst case update time, where $\alpha$ is the arboricity of the graph when receiving an edge-update.

%TODO: Should we add more literature of applications of orientations? They seem to be important in other models as well, e.g. distrib parallel etc
%TODO: How about applications to colouring? Does Solomon+Wein use out-or to colour dynamic graphs? 

\paragraph{Arboricity decomposition}
While an arboricity decomposition of a graph; a division of its edges into minimally few forests; is conseptually easy to understand, computing an arboricity decomposition is surprisingly nontrivial. Even computing it exactly has received much attention~\cite{GabowW,Gabow95,Edmonds1965MinimumPO,PicardQueyranne82}. 
The state-of-the-art for computing an exact arboricity decomposition runs in $\tilde{O}(m^{3/2})$ time \cite{GabowW,Gabow95}.
In terms of not-exact algorithms there is a 2-approximation algorithm~\cite{ArikatiMZ97,Eppstein94} as well as an algorithm for computing an $\alpha+2$ arboricity decomposition in near-linear time~\cite{blumenstock2019constructive}.

For dynamic arboricity decomposition, Bannerjee et al.~\cite{Banerjee} give a dynamic algorithm for maintaining the current arboricity. The algorithm has a near-linear update time. They also provide a lower bound of $\Omega(\log{n})$. % for dynamically maintaining arboricity.
Henzinger Neumann Wiese~\cite{henzinger2020explicit} provide an $O(\alpha)$ arboricity decomposition in $O(\poly(\log n , \alpha))$ time; their result also goes via out-orientation, and they provide a dynamic algorithm for maintaining a $2\alpha'$ arboricity decomposition, given access to any black box dynamic $\alpha'$ out-degree orientation algorithm. 
Most recently, there are algorithms for maintaining $(\alpha+2)$ forests in $O(\operatorname{poly} (\log(n) , \alpha))$ update-time~\cite{christiansenICALP}, and $(\alpha + 1)$ forests in $\tilde{O}(n^{3/4}\operatorname{poly}(\alpha))$ time~\cite{christiansenMFCS}.

\input{introtable}

%s\newpage

\subsection{Preliminaries and Parameterization of the Algorithm}\label{sec:prelim}
Let $G = (V, E)$ be an undirected graph with $n$ vertices and $m$ edges. 
We study $G$ subject to edge insertions and deletions.
The \emph{density} of a graph $G$ is defined as $\rho := \frac{|E[G]|}{|V(G)|}$, and the \emph{maximum subgraph density} of $G$ is then the density of the densest subgraph of $G$. A closely related measure of uniform sparsity is the \emph{arboricity} of a graph, defined as:
\[
\alpha := \max \limits_{H \subseteq G , |V(H)| \geq 2} \left \lceil \frac{|E(H)|}{|V(H)|-1} \right \rceil
\]

\noindent
For a graph $G$, denote by $\overrightarrow{G}$ an \emph{orientation} of $G$: a version of $G$ which imposes upon each edge of $G$ a direction.
For any vertex $u \in V$, we subsequently denote by $N^+(u)$ all vertices $w$ with $\overrightarrow{uw} \in \overrightarrow{G}$ and by $N^-(u)$ all vertices $w$ with $\overrightarrow{wu} \in \overrightarrow{G}$.
We denote the out-degree $d^+(u)$ as all edges directed from $u$, and the in-degree $d^-(u)$ as all edges directed towards $u$. Observe that whenever $G$ contains multiple edges between the same two vertices, then $d^+(u)$ can be larger than $|N^+(u)|$. 
The maximum out-degree of $\overrightarrow{G}$ is defined as: $\Delta(\overrightarrow{G}) := \max_{v \in V} d^+(v)$. It was shown by Picard \& Queyranne~\cite{PicardQueyranne82} that $\lceil \rho \rceil = \min_{\overrightarrow{G}} \Delta(\overrightarrow{G})$, and so it follows that $\rho \leq \Delta(\overrightarrow{G})$ always.

%\todo[inline]{end of preliminaries}

%\[
%\overrightarrow{\Delta_G} :=  \min_{\overrightarrow{G}} \Delta(\overrightarrow{G})  \textnormal{ \hspace{1cm} through duality (\cite{sawlani2020near}) we know: \hspace{1cm} } \rho \leq \overrightarrow{\Delta_G}
%\]

\paragraph{Simultaneous independent parallel work} by Chekuri and Quanrud on subgraph density estimation~\cite{ChekuriKent} 
contains overlapping techniques with the paper at hand. 
%uses techniques overlapping with this paper.

%\todo{E: is this short enough?}
%, and these techniques were discovered independently in parallel. 
%There are also technical contributions present in our paper not occurring in \cite{ChekuriKent}, and thus we obtain faster worst-case running times and a wider range of applications. 
%\chr{I think it is fine}

\section{Overview of techniques}\label{sec:tech}
An important theoretical insight that helps shed light on the topic of this work, is the following linear program for minimizing the maximal fractional out-degree of an edge orientation:
%a dual pair of linear programs. Namely that for determining the density of the densest subgraph, $\rho$, and that for minimising the maximal fractional out-degree of an edge-orientation: 
%
%The starting point of our work is the linear program for determining the densest subgraph (with density $\rho$)\todo{Can I inject this? -I} and its dual linear which corresponds to the fractional orientation problem:
%\begin{align*}
%   \textbf{Densest Subgraph} & &  & & \textnormal{~} & \hspace{4.5em} & \textbf{Fractional Orientation} & & \\
%   \max \sum_{\{u,v\}\in E} y_{u,v} & & \text{s.t.} & & \textnormal{~} & & \min \rho & &\text{s.t.}\\
%x_u,x_v \geq y_{u,v}  & & \forall u,v\in V, \{u,v\}\in E & & \textnormal{~} & & \alpha_{u,v} + %\alpha_{v,u} \geq 1 & &\forall \{u,v\}\in E \\
%\sum_{u\in V} x_u  \leq 1 & & & & \textnormal{~} & & \rho \geq \sum_{v\in V} \alpha_{u,v} & &\forall u\in V \\
%x,y \geq 0 & & & & \textnormal{~} & & \rho,\alpha \geq 0 & &
%\end{align*}
%
\begin{align*}
\min \rho & ~~~~\text{ s.t.}\\
\alpha_{u,v} + \alpha_{v,u} \geq 1 & ~~~~\forall \{u,v\}\in E \\
\rho \geq \sum_{v\in V} \alpha_{u,v}  &~~~~~\forall u\in V \\
\rho,\alpha \geq 0 & 
\end{align*}

\noindent
%Any optimal solution of the primal LP is a convex combination of normalized indicator variables inducing a densest subgraph. That is, if $S\subset V$ is the unique densest subgraph, $x_u = \frac{1}{|S|}$ if $u\in S$ and $x_u=0$ if $u\notin S$. 
The LP induces a (fractional) orientation of the edges. Specifically, $\alpha_{u,v}=1$ implies that the edge is oriented from $u$ to $v$. Similarly, we define $fd^+(u):=\sum_{v\in V} \alpha_{u,v}$ to be the fractional out-degree of the node $u$.
%By duality, the maximum fractional out-degree is equal to the density of the densest subgraph. Similarly, this relates the %maximum outdegree of the best possible out-orientation to the density $\rho$. Specifically, 
There always exist an out-orientation with maximum outdegree $\lceil \rho \rceil$. Note that by duality, the optimal value of $\rho$ is also equal to the density of the densest subgraph.

Our goal is to (i) maintain a fractional out-orientation and (ii) round the fractional out-orientation with low deterministic update time per update.

\paragraph{Maintaining $\rho$ with Small Recourse.}
 Let $G_1,\ldots, G_t$ be a sequence of graphs where $G_i$ and $G_{i+1}$ differ in precisely one edge. Our first challenge is to find a sequence of feasible dual solutions $D_i=\begin{pmatrix}\rho_i \\ \alpha_i\end{pmatrix}$ such that the number of changes $\|\alpha_i-\alpha_{i+1}\|_0$  following an edge update are bounded, while ensuring that the density $\rho_i$ is close to the optimal density of $G_i$. 
Suppose that we are given an upper bound $\rho_{\max}$ on the minimal maximum outdegree for every $G_i$ . We relax the constraint $\rho \geq \sum_{v\in V} \alpha_{u,v}$ to:
\begin{equation}
\label{eq:roughconstraint}
 (1+\varepsilon)\cdot \rho_{\max} \geq \sum_{v\in V} \alpha_{u,v} 
%\label{eq:roughconstraint}
\end{equation}
 Clearly, if the $i$th edge  update deletes an edge $\{u,v\}$, then the fractional solution $D_{i+1}$ (obtained by changing only $\alpha_{u,v}=\alpha_{v,u} \gets 0$) is still feasible. However, if an edge $\{u,v\}$ is inserted, for any choice of $\alpha_{u,v}+\alpha_{v,u}\geq 1$, either $fd^+(u)$ or $fd^+(v)$ may be greater than $(1+\varepsilon)\rho_{\max}$. To illustrate the subsequent ideas, we assume that the fractional orientation is integral (that is $\alpha_{a,b}$ is either $1$ or $0$ for any edge $\{a,b\}$) 
 %and that $\alpha_{u,v}=1$. 
 For any node $u$, let $V_j(u)$ be the set of nodes in the directed graph given by the orientation 
 %induced by $\alpha$ 
 at distance $j$ from $u$. Let $k$ be the smallest value such that there exists a $w\in V_k(u)$ with $fd^+(w) < (1+\varepsilon)\rho_{\max} -1$. Suppose we re-orient all edges on the path from $u$ to $w$. The outdegree of every interior node along this path does not change. $fd^+(u)$ decreases by $1$ and $fd^+(w)$ increases by $1$, so we once again have a feasible solution. Since $\rho$ is assumed to be bounded by $\rho_{\max}$, this implies that the size of the sets $V_k$ grow exponentially, i.e. $|V_k|>(1+\varepsilon)^k$. Thus, the length of the path from $u$ to $w$ is at most $O(\varepsilon^{-1}\log n)$. A similar argument can be made in the case of fractional orientations. Essentially, the main idea is to fix the fractional values of $\alpha_{u,v}$ to be $\{0,\frac{1}{b}, \frac{2}{b},\frac{3}{b},\ldots ,1\}$ for some positive integer $b\in \polylog(n)$ and simulate a fractional orientation by duplicating an edge $b$ times, rescaling $\rho_{\max}$ by $b$ and inserting each duplication. Effectively, this shows that there always exist an update to the dual variables with bounded recourse, while retaining a nearly optimal out-orientation (and thus density estimation). What remains to be shown is that such an update can be computed quickly.

\paragraph{From Small Recourse to Small Update Time.}
Starting with an arbitrary feasible dual, or even an optimal, but otherwise arbitrary dual, it is still difficult to find $w$ without scanning the entire graph. Thus, a natural invariant is to impose the condition $d^+(u) \leq d^+(v) + c$ for some appropriately chosen value of $c$. Kopelowitz, Krauthgamer, Porat and Solomon \cite{KopelowitzKPS13} use $c=O(1)$. Unfortunately, this does not bound the recourse as our approach will. Thus, this approach immediately leads to an update time of at least $O(\rho)$. Sawlani and Wang \cite{sawlani2020near} avoid this by imposing the condition:
\begin{equation}
    \label{eq:additivelarge}
    d^+(u) \leq d^+(v) + \epsilon \rho_{\max},
\end{equation}
for any two neighboring nodes $u$ and $v$, instead of Equation \ref{eq:roughconstraint}. This ensures that the outdegree between neighbors cannot be too different. Indeed, whenever we process an edge, re-imposing this condition triggers a sequence of outdegree modifications that balance the outdegrees automatically throughout the graph. Specifically, if after insertion of $\{u,v\}$: $d^+(u)>(1+\varepsilon)\cdot \rho_{\max}$, the key structural implication of Equation \ref{eq:additivelarge} is that there exists an out-neighbor $w$ of $u$ with $d^+(w) \leq \rho_{\max}-1$. It is instructive to note that a similar condition cannot be enforced when using \ref{eq:roughconstraint}. Of course, merely flipping the orientation of $\{u,w\}$ might lead to new violations of Equation \ref{eq:additivelarge}, but the argument of always being able to find a low-outdegree neighbor can be applied inductively, ensuring that the overall number of edge re-orientations remains $O(\varepsilon^{-1}\log n)$.
The more straightforward update time requires scanning all out-neighbors for every node along the path of edge re-orienations, resulting in an update time of $O(\rho_{\max} \cdot \varepsilon^{-1} \log n)$ (see for example Kopelowitz, Krauthgamer, Porat and Solomon \cite{KopelowitzKPS13}).
If $\rho_{\max}$ is in $O(\log n)$, we retain an acceptable update time. 
If $\rho_{\max}$ is large, Sawlani and Wang \cite{sawlani2020near} show how to replace the dependency on $\rho_{\max}$ to $\poly(\varepsilon^{-1},\log n)$ by showing that (i) accurate estimates for the out-degree are sufficient for the analysis to go through and (ii) that the estimates can be maintained by lazily informing the out-neighbors in a round robin fashion, see also \cite{BernsteinS15,GSSU22} for similar ideas to inform neighbors of approximate degrees.

\paragraph{From Density to Maintaining Out-Orientations}
Unfortunately, using Equation \ref{eq:additivelarge} leads to additional problems. If $\rho$ is substantially smaller than $\rho_{\max}$, the approximation factor becomes poor. Secondly (and arguably more importantly) if $\rho$ is substantially larger than $\rho_{\max}$, edges can no longer be inserted such that the constraints can be satisfied. 
Sawlani and Wang \cite{sawlani2020near} handle this by selecting multiple estimates of $\rho_{\max}$ and attempting to find feasible solutions for such estimate. If for some choice of $\rho_{\max}$, Equation \ref{eq:additivelarge} cannot be satisfied any longer, subsequent edge insertions are postponed and only processed once $\rho$ decreases again. This leads to additional overhead in the running time. But most importantly, it also means that when switching between different estimates of $\rho$ at the $i$th edge update, $\|\alpha_i-\alpha_{i+1}\|_0$ is no longer bounded. Thus, while we are able to maintain an estimate for $\rho$ in this way, maintaining an edge orientation is not possible.
Thus the key new insight of our work is to impose the condition:
\begin{equation}
    \label{eq:inv1}
    d^+(u) \leq d^+(v) (1+ \varepsilon) + \theta
\end{equation}
for any constant $\theta \geq 0$.
With this condition we are able to maintain a single dual solution that has bounded recourse; even when $\rho$ changes.
This naturally avoids coordinating between different estimates of $\rho$.
As a consequence, we maintain a (nearly optimal) fractional edge orientation in worst-case deterministic update time.
To round the fractional orientation to an integral $2$-approximate one, we simply direct an edge $\{u,v\}$ towards $v$ if $\alpha_{u,v} \geq 0.5$ \footnote{If $\alpha_{u,v}=0.5$, this would imply that the edge is directed towards both $u$ and $v$. This does not affect the bound on the maximum out-degree. In case a true orientation is desired, one can break ties arbitrarily.}.
Better approximation factors can be obtained by combining our techniques with recent work of \cite{christiansenICALP}, see Section \ref{sec:onepluseps}.

\paragraph{Further Trade-Offs and Applications}
Maintaining a nearly optimal edge orientation is interesting in its own right. Nevertheless, we obtain further results. First, by leveraging the duality between the densest subgraph problem and edge orientations, we are able to maintain a $(1+\varepsilon)$-approximate densest subgraph, improving on previous work by Sawlani and Wang \cite{sawlani2020near}.

Second, out-orientations are frequently used as a data structure subroutine in other algorithms. 
%Previously, efficient arboricity-sensitive worst-case algorithms required knowing an upper bound $\alpha_{\max}$ on the arboricity at any time.
For example, maintaining an $\beta$ out-degree orientation with worst-case update time $T$ immediately gives rise to an $O(\beta + T)$ update time for maintaining a maximal matching. By tweaking the parameters of our algorithms and in particular with a careful choice of the invariant in Eq. \ref{eq:inv1}, we can obtain different trade-offs between update time $T$ and outdegree. Specifically, we decrease $T$ to $O(\log n \log \alpha)$, while increasing the maximum out-degree to $O(\alpha + \log n)$. %Note that f
For maximal matching, since our update time $T$ is at least $O(\log n)$, the weaker guarantee on the out-degree is preferable. 
%the overall update time is $O(\alpha + \min(\alpha\log n \log \alpha, \log^3 n \log \alpha))$. 
%In the low-arboricity setting, this is comparable (though always slightly better) to existing update times of $O(\alpha^2 + \alpha \log n)$ by Kopelowitz, Krauthgamer, Porat, and Solomon \cite{kopelowitz2014orienting}.
%For higher values of $\alpha$ this improves over both Kopelowitz, Krauthgamer, Porat, and Solomon as well as the $O(\sqrt{m})$ algorithm by Neiman and Solomon \cite{NeimanS16}.

%\todo[inline]{ER: I want to take the part of section two starting with "we now introduce several" and move it \emph{here} under the headline parameteri...tion of the algorithm. Objections? }

\subsection{Parameterisation of the Algorithm}
We now introduce several components of the algorithm and analysis that can be specified to obtain various trade-offs between quality of the out-orientation and update time.
Our algorithms have the three main parameters: $\eta,\gamma>0$ and a positive integer $b$. 
Recall that in Section~\ref{sec:tech} we explained that we can maintain a fractional orientation. 
The role of $b$ corresponds to the number of times we duplicate every edge such that the orientation that we maintain is integral at all times.
Specifically, we maintain a graph $G^b$ where every edge is duplicated $b$ times, and an out-orientation $\overrightarrow{G}^b$ over this graph. 
The orientation $\overrightarrow{G}$ is subsequently obtained by rounding the corresponding fractional orientation.
The role of $\eta$ is inherent in the definition of the following invariants:
\begin{invariant}
\label{inv:degrees}
At all times, we maintain an orientation $\overrightarrow{G}^b$ where for every directed edge $\overrightarrow{uv}$ in $\overrightarrow{G}^b$: 
\[
d^+(u) \leq (1 + \eta \cdot b^{-1}) \cdot d^+(v). 
\]
\end{invariant}

\noindent
For $\eta$ small enough, we note that the Invariant~\ref{inv:degrees} cannot be satisfied in general\footnote{As a simple example, a star will require some node to have an outdegree of $0$.}, unless the number of duplicate edges $b$ is large enough. However, the number of duplicate edges required reflects poorly upon our running time. 
An alternative way of guaranteeing satisfiability is the following:
\begin{invariant}
\label{inv:degrees_additive}
At all times, we maintain an orientation $\overrightarrow{G}^b$ where for every directed edge $\overrightarrow{uv}$ in $\overrightarrow{G}^b$: 
\[
d^+(u) \leq (1 + \eta \cdot b^{-1}) \cdot d^+(v) +2 . 
\]
\end{invariant}

\noindent
Throughout the paper, we denote $\theta = 0$ if we are maintaining Invariant~\ref{inv:degrees} and $\theta = 1$ otherwise.
This way, we can write that we maintain: $d^+(u) \leq (1 + \eta \cdot b^{-1}) \cdot d^+(v) + 2\theta$ instead.
The tighter the inequalities are, the closer the maximum out-degree of the maintained out-orientation is to the arboricity.
Hence, setting $\theta = 0$ will give a better approximation than $\theta = 1$.

\newpage
\section{A Structural Theorem}\label{sec:struc}

In this section, we formally establish the relationship between maintaining Invariant~\ref{inv:degrees} or \ref{inv:degrees_additive}, and estimating density and arboricity of the graph.

\begin{theorem}
\label{thm:structural}
Let $G$ be a graph and let $G^b$ be $G$ with each edge duplicated $b$ times. Let $\rho_b$ be the maximum subgraph density of $G^b$.
Let  $\overrightarrow{G}^b$ be any orientation of $G^b$ which has the following invariant: for every directed edge $\overrightarrow{uv}$ it must be that $d^+(u) \leq (1 + \eta \cdot b^{-1}) \cdot d^+(v)+c$ for some constant $c$. % \newline 
Then for any constant $\gamma > 0$ 
%and $c > 0$ 
there exists a value $k_{\max} \leq \log_{1 + \gamma} n$ for which:
\[
(1+\eta\cdot b^{-1})^{-k_{\max}}\Delta(\overrightarrow{G}^b) \leq (1+\gamma)\rho_b +c(\eta^{-1}\cdot b+1).
\]
\end{theorem}
\begin{proof}
We define for non-negative integers $i$ the sets: 
\[
T_i := \left\{ v \Bigm| d^+(v) \geq \Delta(\overrightarrow{G}^b) \cdot \left(1 + \eta\cdot b^{-1}\right)^{-i} - c\sum_{j = 1}^i \left(1 + \eta\cdot b^{-1}\right)^{-j} \right\}
\]

\noindent
Observe that for all non-negative integers $i < j$, $T_i \subseteq T_j$.
Moreover, observe that $T_1$ contains at least one element (the element of $\overrightarrow{G}^b$ with maximal out-degree), and each $T_i$ at most $n$ elements (since they can contain at most all vertices of $G$). 
Let $k$ be the smallest integer such that $|T_{k+1}| < (1 + \gamma) |T_k|$.
It follows that $k$ is upper bounded by the value $k_{max} = \log_{(1 + \gamma)} n$.

In order to bound the maximum out-degree of $\overrightarrow{G}^b$, we want to show that no edges can be oriented from $T_{k}$ to a vertex not in $T_{k+1}$. To do so, we assume two such candidates $u \in T_k$ and $v \not \in T_{k+1}$, and show that $\overrightarrow{uv}$ violates: $d^+(u) \leq (1 + \eta \cdot b^{-1}) \cdot d^+(v)+c$. 
Per assumption we have:
\[
d^+(v) < (1 + \eta\cdot b^{-1})^{-1} (1 + \eta\cdot b^{-1})^{-k} \Delta(\overrightarrow{G}^b) -  c\sum_{j = 1}^{k+1} (1 + \eta\cdot b^{-1})^{-j} 
\]
and
\[
d^+(u) \geq (1 + \eta\cdot b^{-1})^{-k} \Delta(\overrightarrow{G}^b)  - c\sum_{j = 1}^{k} (1 + \eta\cdot b^{-1})^{-j}
\]
It follows that
\begin{align*}
    (1 + \eta\cdot b^{-1}) d^+(v) + c &< (1 + \eta\cdot b^{-1})^{-k} \Delta(\overrightarrow{G}^b) - c\sum_{j = 0}^{k} (1 + \eta\cdot b^{-1})^{-j} + c \\
    &\leq (1 + \eta\cdot b^{-1})^{-k} \Delta(\overrightarrow{G}^b)  - c\sum_{j = 1}^{k} (1 + \eta\cdot b^{-1})^{-j} \\
    &\leq d^+(u)
\end{align*}
This would violate the assumed invariant of $\overrightarrow{G}^b$.
Hence for any $u \in T_k$ and any edge $\overrightarrow{uv}$, we have $v \in T_{k+1}$ and thus: 
$
\sum_{u \in T_k} d^+(u) \leq |E[T_{k+1}]|.
$
Finally, we can bound the density as:
\[
\rho_b \geq \rho(T_{k+1}) = \frac{|E[T_{k+1}]|}{|T_{k+1}|} \geq \frac{\sum_{u \in T_k} d^+(u) }{(1+\gamma) |T_k| } \geq \frac{|T_k| \cdot \left( \left(1 + \eta\cdot b^{-1}\right)^{-k} \Delta\left(\overrightarrow{G}^b\right) - c\sum_{j = 1}^k \left(1 + \eta\cdot b^{-1}\right)^{-j} \right) }{(1+\gamma)|T_k|} 
\]
As a result we find:
\[
\left(1+\gamma\right)\rho_b +\frac{c}{1 - \frac{1}{1 + \eta\cdot b^{-1}}}  \geq  \left(1 + \eta\cdot{} b^{-1}\right)^{-k_{\max}} \Delta\left(\overrightarrow{G}^b\right) 
\]
which concludes the proof.\end{proof}

\newpage

%\subsection{technical stuff}
%\todo[inline]{This can be removed now I think. I kept it around in case someone wants it.}
%\todo[inline]{Let's move the following to the next section.}
%This is formalized through the following recursively defined nested node sets.
%Given an orientation $\overrightarrow{G}$ and a function $h: \mathbb{N} \to \mathbb{R}%_{\geq 0}
%$, we define for non-negative integers $i$ 
%(and some function $h(i)$) 
%the set:
%\[
%T_i:=\{v\in V~|~d^+(v) \geq \left(1+\eta\cdot b^{-1} \right)^{-i}\cdot \Delta(\overrightarrow{G}) + h(i)\}
%\]
%\noindent
%Finally, let $k$ be the smallest value such that
%\begin{equation}
%\label{eq:densestsubgraph}
% |T_{k+1}| < (1+\gamma)\cdot |T_k|,   
%\end{equation}
%we denote by $k_{\max}$ be an upper bound on $k$. Note that $k_{\max} \leq \log_{1+\gamma} n$.\todo{E: any $h$? Does this require $T_i$ is nonempty for before $k_{1+\gamma} n$? E2: what if $h=-\infty$?}\chr{I think the definition of $k_{\max}$ requires no assumption on $h$; this will always hold.}

%\todo[inline]{If we only ever use the same function $h$, then I think we should define it here before $k_{\max}$. -e}

\newpage

\section{A Simple Algorithm for Maintaining the Invariants}
\label{sec:basic}
We show how to maintain Invariant~\ref{inv:degrees} or \ref{inv:degrees_additive} for the graph $\overrightarrow{G}^b$ in two steps: here, we give a very simple algorithm with update time $O(\rho \log \rho \cdot \polylog(n))$ (where $\rho$ is the maximal subgraph density). The subsequent section then improves the running time.
We present an algorithm to dynamically maintain either Invariant~\ref{inv:degrees} or~\ref{inv:degrees_additive} (i.e. we maintain one chosen invariant through the same algorithm). 
Let $G^b$ be the graph $G$ with edges duplicated $b$ times (where we want to maintain the invariant). 
For every vertex $u$ we store: 
\begin{enumerate}[(a), noitemsep, nolistsep]
    \item The value $d^+(u)$ of the current orientation $\overrightarrow{G}^b$,
    \item The set $N^+(u)$ in arbitrary order, and
    \item The set $N^-(u)$ in a doubly linked list of non-empty buckets $B_j$ sorted on $j$.
    Each $B_j$ contains all $w \in N^-(u)$ with $d^+(w) = j$ as a doubly linked list in arbitrary order.
\end{enumerate}

\noindent
Finally, we distinguish between an insertion and deletion (which adds or removes a directed edge; possibly altering some $d^+(u)$ and recursing) and the \textbf{add} and \textbf{remove} operations (which adds or removes a directed edge to offset a recent insertion or deletion).
In our recursive algorithm we assume that for any edge insertion $(u, v)$ in $G^b$, we call $\textnormal{Insert}(\overrightarrow{uv})$ whenever $d^+(u) \leq d^+(v)$. 
Recall that $\theta$ indicates we maintain Invariant~\ref{inv:degrees}
or \ref{inv:degrees_additive} and that $\eta$ and $b$ are parameters set beforehand.

\renewcommand\algorithmicthen{}

\noindent\begin{minipage}[t]{.5 \textwidth}
\null 
 \begin{algorithm}[H]
    \caption{Insertion}
    \label{alg:insertion}
    \begin{algorithmic}
      \STATE $x \gets arg \, \min \{ d^+(w) \mid w \in N^+(u) \}$ 
      \IF{ 
      $
        d^+(u) + 1 > (1 + \eta b^{-1}) \cdot d^+(x) + 2 \theta
      $
     }
      \STATE Remove$(\overrightarrow{ux})$
      \STATE \emph{Insert($\overrightarrow{xu}$)}
      \ELSE
      \STATE $d^+(u) \gets d^+(u) + 1$
      \FORALL{$w \in N^+(u)$}
      \STATE Update $d^+(u)$ in LinkedList($N^-(w)$) 
      \ENDFOR
      \ENDIF
    \end{algorithmic}
  \end{algorithm}
\end{minipage}~%
\begin{minipage}[t]{.5\textwidth}
\null
 \begin{algorithm}[H]
    \caption{Deletion}
    \label{alg:deletion}
    \begin{algorithmic}
      \STATE $x \gets \textnormal{First( First( LinkedList of $ N^-(u)$))} $
      \IF{
       $d^+(x) > (1 + \eta b^{-1}) \cdot ( d^+(u) - 1)  + 2 \theta
      $
}
            \STATE Add$(\overrightarrow{ux})$
            \STATE \emph{Delete($\overrightarrow{xu}$})
      \ELSE
      \STATE $d^+(u) \gets d^+(u) - 1$
      \FORALL{$w \in N^+(u)$}
      \STATE Update $d^+(u)$ in LinkedList($N^-(w)$) 
      \ENDFOR
      \ENDIF
    \end{algorithmic}
  \end{algorithm}
\end{minipage}
    
\subsection{Maintaining Invariant~\ref{inv:degrees_additive}.}
We show that we can dynamically maintain Invariant~\ref{inv:degrees_additive} through the following theorem:
    
\begin{theorem}
\label{thm:invariant2}
Let $G$ be a graph subject to edge insertions and deletions.
Choose $c = 2$, $b = 1$,  $\eta = \frac{1}{2 \log_e n}$, and $\gamma > 0$ constant such that $\eta b^{-1} < \frac{1}{\log n \max \{  \gamma^{-1}, 1\}}$.
We can choose $\gamma$ to maintain an out-orientation $\overrightarrow{G}^b = \overrightarrow{G}$ in $O\left( (\rho + \log n) \cdot   \log n \cdot \log \rho \right)$ time per operation in $G$ such that  Invariant~\ref{inv:degrees_additive} holds for $\overrightarrow{G}$ and:
\begin{itemize}[noitemsep, nolistsep]
 \item $\forall u$, the out-degree $d^+(u)$ in $\overrightarrow{G}$ is at most $O( \rho + \log n)$, \quad \quad \small{(i.e. $\Delta( \overrightarrow{G}) \in O( \rho + \log n)$)}
\end{itemize}
where $\rho$ is the density of $G$ at the time of the update.
\end{theorem}

\begin{proof}
Invariant~\ref{inv:degrees_additive} demands that $\forall \overrightarrow{uv}$ we maintain $d^+(u) < (1 + \eta b^{-1}) d^+(v) + 2$.
First, we prove that the invariant gives us the desired upper bound on the out-degree of each vertex. 
Then, we show that we can maintain the invariant in three steps where we prove: 

\begin{description}[noitemsep, nolistsep]
    \item[Correctness:]  for all $G^b$, there exists an orientation $\overrightarrow{G}^b$ such that Invariant~\ref{inv:degrees_additive} holds,
        \item[Recursions:] the algorithm recurses at most $\left( \eta b^{-1} \right)^{-1} \cdot \log \left( \rho \right)$ times to realize $\overrightarrow{G}^b$, and 
    \item[Time:]  the algorithm spends $O(\rho + \log n)$ time before entering the recursion.
\end{description}
\noindent
We prove these three properties for insertions only. The proof for deletions is symmetrical.

\subparagraph{Invariant~\ref{inv:degrees_additive} implications.}
By choice of $\gamma$, we have $k_{\max} \leq \frac{ \log_e n}{  \log_e (1+\gamma)}$ (proof of Theorem~\ref{thm:structural}):

\noindent $
    (1+\eta\cdot b^{-1} )^{-k_{\max}} = \exp(-\log_e(1+\eta b^{-1} )\cdot k_{\max}) \geq  \exp( - ( \eta b^{-1} ) \cdot k_{\max}) \Rightarrow \textnormal{ \tiny  (Using $ x \geq \log(1 + x)$)}$

\noindent $
(1+\eta\cdot b^{-1} )^{-k_{\max}}   \geq  \exp( - ( 2 \log_e n )^{-1}  \cdot \frac{\log_e n}{\log_e (1 + \gamma)}) =  \exp( - \frac{2 }{\log_e (1 + \gamma)}) \geq \Omega(1) \textnormal{ \tiny (noting that $\gamma$ constant)}  $

\noindent
We apply Theorem~\ref{thm:structural} (setting $c = 2$) to obtain:
$(1+\eta\cdot b^{-1})^{-k_{\max}}\Delta(\overrightarrow{G}) \leq (1+\gamma)\rho +c(\eta^{-1}\cdot b+1)$.
This, together with our inequality, implies that  $\Delta(\overrightarrow{G})\in O(\rho+\log n)$.
%%$\Delta(\overrightarrow{G})\in  O(\rho + \log n)$%
%$\Omega(\rho - \log n) \leq \rho \leq \Delta(\overrightarrow{G})$.

\subparagraph{Correctness.}
For any $\overrightarrow{G}^b$, for any edge $(u, v)$, we may always orient the edge from either $u$ to $v$ or $v$ to $u$. 
Indeed, suppose that we cannot insert $\overrightarrow{vu}$ (i.e.  that $d^+(v) + 1 > (1 + \eta b^{-1}) d^+(u) + 2$). Then since $(1 + \eta b^{-1}) > 1$ 
we may insert $\overrightarrow{uv}$ without violating Invariant~\ref{inv:degrees_additive} between $u$ and $v$.
This implies that for every graph $G$, there exists a graph $\overrightarrow{G}^b$ for which Invariant~\ref{inv:degrees_additive} holds. 

What remains is to show that we can realize this graph.
For any graph $G$, we claim that we can create a directed graph $\overrightarrow{G}$ which does not violate Invariant~\ref{inv:degrees_additive} by inserting (and orienting) all the edges of $G$ one by one, and greedily flipping edges which violate the invariant. 
To prove the correctness of our approach, we prove that greedily flipping violating edges after an insertion always terminates.
First note that whenever our greedy algorithm flips an edge $\overrightarrow{xy}$, we decrease $d^+(x)$ to restore it to the value that it was before the insertion of $\overrightarrow{uv}$.
Hence afterwards there exists no edge $\overrightarrow{zx}$ or $\overrightarrow{xz}$ which can violate Invariant~\ref{inv:degrees_additive}. 
Flipping $\overrightarrow{xy}$ increases $d^+(y)$ by one, and thus afterwards we may violate Invariant~\ref{inv:degrees_additive} only for edges $\overrightarrow{yz}$.
We greedily select one such edge $\overrightarrow{yz}$ to flip, and recurse. 
Whenever this occurs, it must be that before inserting $\overrightarrow{uv}$: $d^+(y) > (1 + \eta b^{-1}) d^+(z) + 1$ and thus $d^+(y) > d^+(z) + 1$. 
This implies that our algorithm always greedily flips a `directed chain' of edges which never visits a vertex twice (thus, our algorithm terminates). 
We upper bound its length:

\subparagraph{Recursions.}
Suppose that we insert an edge $\overrightarrow{uv}$, and that we recursively flip a directed chain of edges of length $T$ to reach a vertex $x_T$. 
We now make a case distinction:

Let $\rho < \log n$. The above analysis shows that: the out-degree $d^+(u)$ is at most $O(\log n)$ and for any two consecutive vertices $(y, z)$ on the chain: $d^+(z) < d^+(y)$. Thus, $T \in O(\log n)$.

Let $\rho \geq \log n$. The above analysis shows that: out-degree $d^+(u)$ is at most $O(\rho)$ and for any two consecutive vertices $(y, z)$ on the chain: $d^+(y) > (1 + \eta b^{-1}) d^+(z)$. 
It follows that $d^+(u) > (1 + \eta b^{-1})^T \cdot d^+(x_T)$. 
We can choose $\gamma$ such that $\eta b^{-1}$ is sufficiently small (and use the Taylor expansion of $\frac{1}{\log (1 + \eta b^{-1})}$) such that the number $T$ of recursions is upper bound by:
\[
T \leq \frac{ O( \log  \rho)  }{\log(1 + \eta b^{-1} )} \leq 2 \frac{ O( \log \rho ) }{\eta b^{-1}} = O \left(  \left( \eta b^{-1} \right)^{-1} \cdot \log \rho   \right) = O\left( \log n \log \rho \right).
\]

\subparagraph{Time spent (before recursion).}
Whenever we insert $\overrightarrow{uv}$, we increase $d^+(u)$. 
Increasing $d^+(u)$ may cause us to violate Invariant~\ref{inv:degrees_additive} for some edge $\overrightarrow{uw}$ for $w \in N^+(u)$.
Whenever this is the case, it must be the case for $x \gets \arg \min \{ d^+(w) \mid w \in N^+(u) \}$.
We obtain $x$ in $O(\rho + \log n)$ time by checking all out-edges of $u$. 
If the invariant is violated, we recurse. 
If the invariant is not violated, we update for each $w \in N^+(u)$, the  corresponding linked lists for $w$. 
In this structure, the vertex $u$ is moved from the bucket $B_{d^+(u)}$ to the bucket $B_{d^+(u) + 1}$ and thus this move can be performed in $O(1)$ time. 
It follows that for each insertion, we spend at most $O( \rho + \log n)$ time before we recurse.  

\subparagraph{Conclusion.}
Our algorithm spends $O( \rho + \log n)$ before recursion, and recurses $O( \log n \log \rho )$ times. 
The proofs for deletions are identical. In this case, the vertex $x$ can be found in $O(1)$ time instead (by selecting for the first bucket, the first element, in $O(1)$ time).
\end{proof}

\subsection{Maintaining Invariant~\ref{inv:degrees}.}
Recall that $G^b$ is the graph $G$ with its edges duplicated $b$ times (see Theorem~\ref{thm:structural}). 
By setting $b \in O(\log n)$, we will guarantee that we can maintain Invariant~\ref{inv:degrees} for $\overrightarrow{G}^b$.
We convert this into an orientation $\overrightarrow{G}$ by choosing $\overrightarrow{uv}$ whenever $\overrightarrow{G}^b$ contains more edges from $u$ to $v$.

\begin{theorem}
\label{thm:invariant1}
Let $G$ be a  graph subject to edge insertions and deletions.
Choose  $\eta > 2$, $\gamma$ constant, and $b$ be such that $\eta b^{-1} < \frac{1}{100 \log n \max \{  \gamma^{-1}, 1\}}$. 
Let $G^b$ be the graph where each edge in $G$ is duplicated $b$ times.  
Algorithms~\ref{alg:insertion} and \ref{alg:deletion}  maintain out-orientations $\overrightarrow{G}^b$ and $\overrightarrow{G}$ s.t.:
\begin{itemize}[noitemsep, nolistsep]
    \item   Invariant~\ref{inv:degrees} holds for $\overrightarrow{G}^b$, 
    \item $\forall v$, the out-degree $d^+(v)$ in $\overrightarrow{G}^b$ is at most $O( b  \cdot \rho)$, and
    \item $\forall u$, the out-degree of $u$ in $\overrightarrow{G}$ is at most $O(\rho)$.
\end{itemize}
We use 
$O\left( \rho \cdot   b^3 \cdot \log \rho \right) = O( \rho \cdot \log^3 n \log \rho)$ time per operation in the original graph $G$.
Here, $\rho$ is the density of $G$ at the time of the update.
\end{theorem}

\begin{proof}
For every edge insertion in $G$, we insert $b$ edges in $G^b$, one by one. 
We use Algorithms~\ref{alg:insertion} and \ref{alg:deletion} on $G^b$ to maintain $\overrightarrow{G}^b$.
We store for every edge $(u, v)$ in $G$, how many edges in $\overrightarrow{G}^b$ point from $u$ to $v$ (and, vice versa). 
Finally, for each edge $(u, v)$ we include the edge $\overrightarrow{uv}$ in $\overrightarrow{G}$ whenever there are more edges from $u$ to $v$ in $\overrightarrow{G}^b$. 
%This way, we maintain both $\overrightarrow{G}$ and $\overrightarrow{G}^b$. 
What remains, is to prove that these digraphs have the desired properties (analogue to the proof of Theorem~\ref{thm:invariant2}).

\subparagraph{Invariant~\ref{inv:degrees} implications.}
By choice of $\gamma$, we have $k_{\max} \leq \log_e n / \log_e (1+\gamma).$ Thus using $\log_e(1+x)\geq x/2$ whenever $x\leq 1$:
$
    (1+\eta\cdot b^{-1} )^{-k_{\max}} = \exp(-\log_e(1+\eta b^{-1} )\cdot k_{\max}) \geq  \exp( - ( \eta b^{-1} ) \cdot k_{\max}) = \Omega(1)$.
We now apply Theorem~\ref{thm:structural} to conclude that $\Omega(1) \Delta( \overrightarrow{G}^b) \leq  (1 + \gamma) \rho_b \leq O( \Delta( \overrightarrow{G}^b ) )$. 
Finally, we note that per definition of subgraph density, $\rho_b = O( b \cdot \rho)$.
For all $\overrightarrow{uv}$ in $\overrightarrow{G}$, there must be at least $\frac{1}{2} b \geq \left( \eta b^{-1} \right)^{-1}$ edges in $\overrightarrow{G}^b$ from $u$ to $v$ (else, $\overrightarrow{G}$ would include the edge $\overrightarrow{vu}$ instead). It immediately follows that the out-degree of $u$ is at most $    d^+(u) \cdot \left( \eta b^{-1} \right) \in O(b^{-1} \cdot \rho_b) = O( \rho)$.

\subparagraph{Correctness.}
For any $\overrightarrow{G}^b$ where  Invariant~\ref{inv:degrees} holds,
for any vertex $v$ incident to at least one edge, $d^+(v)$ is at least $\left( \eta b^{-1} \right)^{-1}$.
It follows that for any $(u, v)$, we may insert either $\overrightarrow{uv}$ or $\overrightarrow{vu}$ without violating Invariant~\ref{inv:degrees}.
Indeed, suppose that we cannot insert $\overrightarrow{vu}$.
Then   $d^+(v) + 1 \geq (1 + \eta b^{-1}) \cdot d^+(u) \geq d^+(u) + 1$.
Thus $d^+(u) \leq (1 + \eta b^{-1}) \cdot d^+(v)$ and we may insert $\overrightarrow{uv}$ without violating~\ref{inv:degrees} between $u$ and $v$. 
What remains is to show that for any graph $G$. we can realize $\overrightarrow{G}^b$. 
Just as in Theorem~\ref{thm:invariant2} we simply insert edges of $G$ one by one (inserting $b$ edges in $\overrightarrow{G}^b$: splitting their direction 50/50). 
Then we greedily flip a directed chain of edges in $\overrightarrow{G}^b$ which violate the invariant until we terminate. 
Consider an edge $\overrightarrow{yz}$ in this chain that must be flipped. Before before inserting $\overrightarrow{uv}$, it must be that $d^+(y) >  (1 + \eta b^{-1} ) \cdot d^+(z) - 1 \Rightarrow d^+(y) \geq (1 + \eta b^{-1}) \cdot d^+(z)$. 
Since $z$ is incident to at least one edge it must have at least $\left( \eta b^{-1} \right)^{-1}$ outgoing edges and thus $d^+(y) \geq d^+(z) + 1$.
This implies that each `directed chain' of flipped edges is finite. Suppose a chain has length $T$:

\subparagraph{Recursions.}
Let $x_T'$ be the second to last vertex of a chain.
Our analysis implies that before inserting $\overrightarrow{uv}$, $d^+(u) > (1 + \eta b^{-1})^T d^+(x_T')$.
$d^+(u)$ is at most $O( \rho \left( \eta b^{-1}\right)^{-1} )$ and $d^+(x_T')$ is at least $O(\left( \eta b^{-1}\right)^{-1} ))$. 
Now $\eta b^{-1} < 1$ implies that $T \in O(b \log \rho) = O(\log n \log \rho)$.

\subparagraph{Time and conclusion.}
Each operation in $G$ triggers $O(b)$ operations in $G^b$. We spend $O( b \rho)$ time per operation in $G^b$ before we recurse. 
We recurse $O(b \log \rho))$ times.
Thus, we spend $O( \rho \cdot  b^3  \cdot \log \rho) = O( \rho \log^3 n \log \rho)$ time per insertion in $G$. 
Deletions are near-identical.
\end{proof}

\section{Improved algorithms}
\label{app:improved}

We adapt the algorithm of Section~\ref{sec:basic} to speed up the running time.
Specifically, we take three steps:
First, instead of maintaining for each vertex $u$, for each in-neighbor $w$ of $u$ the \emph{exact} value $d^+(w)$, we maintain some \emph{perceived} value $d^+_u(w)$ (which we show is sufficiently close to $d^+(w)$). 
%We maintain our invariants using these perceived values.
Second, whenever a vertex $w$ updates for its out-neighbors $u$ the perceived value $d^+_u(w)$, 
it does so in round-robin, informing the next $\lceil \frac{128}{\eta b^{-1}} \rceil$ neighbors of the new degree of $u$.

\noindent\begin{minipage}[t]{.5 \textwidth}
\null 
 \begin{algorithm}[H]
    \caption{Insertion}
    \label{alg:insertion_alphaless}
    \begin{algorithmic}
    \FOR{$x$ in next $\lceil \frac{128}{\eta b^{-1}} \rceil$ neighbors in $N^+(u)$} 
    %\STATE $x \gets \textnormal{First( First( LinkedList of $ N^+(u)$))} $
      \IF{ 
      $
      d^+(u) + 1 > (1 + \eta b^{-1} / 2) \cdot d^+(x) + \theta
      $
     }
      \STATE Remove$(\overrightarrow{ux})$
      \STATE \emph{Insert($\overrightarrow{xu}$)}
      \STATE Return
      \ELSE
      \STATE Update $d^+_x(u)$ in LinkedList$(N^-(x))$
      \ENDIF
      \ENDFOR
      \STATE $d^+(u) \gets d^+(u) + 1$
    \FORALL {$w \in  RoundRobin\left(N^+(u), \lceil \frac{128}{\eta b^{-1}} \rceil \right)$
    }
    \STATE Update $d^+_w(u)$ in LinkedList$(N^-(w))$
    \ENDFOR

    \end{algorithmic}
  \end{algorithm}
\end{minipage}~%
\begin{minipage}[t]{.5\textwidth}
\null
 \begin{algorithm}[H]
    \caption{Deletion}
    \label{alg:deletion_alphaless}
    \begin{algorithmic}
      \STATE $x \gets \textnormal{First( First( LinkedList of $ N^-(u)$))} $
      \IF{
      $
       d^+_u(x) > (1 + \eta b^{-1} / 2) \cdot ( d^+(u) - 1)  +  \theta
      $
}
            \STATE Add$(\overrightarrow{ux})$
            \STATE \emph{Delete($\overrightarrow{xu}$})
      \ELSE
      \STATE $d^+(u) \gets d^+(u) - 1$
    \FORALL {$w \in RoundRobin\left(N^+(u), \lceil \frac{128}{\eta b^{-1}} \rceil \right)$ }
    \STATE Update $d^+_w(u)$ in LinkedList$(N^-(w))$
    \ENDFOR
      \ENDIF
    \end{algorithmic}
  \end{algorithm}
\end{minipage}

\noindent
To facilitate these steps, we alter our data structure. For all vertices $u$ in $\overrightarrow{G}^b$ we store:
\begin{enumerate}[(a), noitemsep, nolistsep]
    \item The \emph{exact} value $d^+(u)$ of the current orientation $\overrightarrow{G}^b$,
    \item The set $N^+(u)$ in a linked list and a pointer to the current position in the linked list. Whenever we change $d^+_u$, we inform at most $\lceil \frac{128}{\eta b^{-1}} \rceil$ successors of the pointer (telling them the new degree of $u$). Any node added to $N^+(u)$ is inserted into the linked list immediately prior to the position of the current counter.

    \item The set $N^-(u)$ in a doubly linked list of buckets $B_j$ sorted by $j$ from high to low.
    Each bucket $B_j$ contains all $w \in N^-(u)$ as a linked list in arbitrary order where:
    $
    w \in B_j \Leftrightarrow 
    d^+_u(w) \in \left[ \left(1 + \frac{\eta b^{-1} }{64} \right)^j, \left(1 + \frac{\eta b^{-1} }{64} \right)^{j+1} \right].
     $
\end{enumerate}

\noindent
As the third step, recall that we maintain a graph $G$ and $G^b$.
Let $\theta \in \{ 0, 1 \}$ indicate whether we aim to maintain Invariant~\ref{inv:degrees} or \ref{inv:degrees_additive}. 
Our algorithm below uses the perceived values in a stricter fashion and will reorient an edge $\overrightarrow{xy}$ 
in $\overrightarrow{G}^b$ if it notices that:
\[
d^+(x) + 1 > (1 + \eta b^{-1} / 2)  \cdot d^+(y) + \theta \quad \textnormal{ and } \quad d^+_y(x) > (1 + \eta b^{-1} / 2) \cdot (d^+(y)-1) + \theta.
\]
The first comparison uses the actual out-degrees of $x$ and $y$, whereas the second uses the perceived out-degree of $x$ at the vertex $y$. 
We prove that this ensures that Invariant~$\theta$ is  maintained for the graph $\overrightarrow{G}^b$ (Lemma~\ref{lemma:invariant_maintenance}). 
Finally, we upper bound the running time and the resulting graph density of $\overrightarrow{G}$ (Theorem~\ref{thm:fast_butworse} + \ref{thm:fastnonadd}).
In the algorithm description, we again distinguish between an insertion and deletion and the \textbf{add} and \textbf{remove} operations.

Note that our algorithm uses two types of round robin schemes: 
during insertions, we iterate over $N^+(u)$, visiting up to $\lceil \frac{128}{\eta b^{-1}} \rceil$ out-neighbors (possibly less). 
At the end of the recursion, the last vertex informs exactly $\lceil \frac{128}{\eta b^{-1}} \rceil$ of its out-neighbors of the new out-degree.

%\todo{This is the extra added explainer -I}\todo{Note that this round robin is different from the updating one. This is a round-robin that updates no data structures and runs every time a vertex is visited during an insertion. The other round robin scheme is independent and works every time a degree is changed}

%\begin{figure}
%\caption{Improved algorithms for insertion and deletion.\todo{Reference Figure~\ref{fig:my_label}? OR not have a figure? Or place the pseudo code some place else?}}
%\label{fig:my_label}
%\end{figure}
% ALEKS: Need to remark that if succesfully find low out-degree neighbour, of course we do not move on from it until its out-degree allows us to.
%\todo{Aleks note-to-self: remember to check that things below hold even as the out-degree increases. Should be ok as new out-neighbours can be checked when they become new out-neighbours and then they can be placed at the back of the round robin queue, so that the time it takes until returns actually scale with the out-degree of the last time we looked at this vertex and not the out-degree when we look at it again. We should probably explain this in some way? Chris: Agreed. I did that.}
\begin{lemma}
\label{lemma:offset_alltimes}
During all steps of Algorithm~\ref{alg:deletion_alphaless}, for all $v$, $\forall u \in N^-(v)$:
\[
d^+(u) < d^+_v(u) + \left( d^+(u) \frac{\eta b^{-1} }{64} \right)
\]
\end{lemma}

\begin{proof}
Consider the time where the algorithm sets the value of $d^+_v(u)$.
Every time the value $d^+(u)$ increases (through Algorithm~\ref{alg:insertion_alphaless}), it updates $\frac{128}{\eta b^{-1}}$ of its neighbors. Note that by moving the counter along the linked list, $v$ will be informed about the new degree prior to any node added to $N^+(u)$ subsequently to $v$.
Hence, the out-degree of $u$ can never be increased by Algorithm~\ref{alg:insertion_alphaless} by $ d^+(u) \frac{\eta b^{-1} }{64}$ or more, without Algorithm~\ref{alg:insertion_alphaless} re-setting the value  $d^+_v(u)$.
\end{proof}

\begin{lemma}
\label{lemma:offset_insertion}
Let for a vertex $u$, $d^+(u)$ be incremented by Algorithm~\ref{alg:insertion_alphaless}. 
Moreover let $\theta = 1$ or $\eta \geq 1280$ and $\eta\cdot{}b^{-1} < 1$.
Then for all $ u \in N^-(v)$:
\[
d^+(u) < d^+_v(u) + \left( d^+_v(u) \frac{\eta b^{-1} }{64} \right) + \theta
\]
\end{lemma}

\begin{proof}
Consider when the algorithm sets the value $d^+_v(u)$.
Every time the value $d^+(u)$ increases, it updates $\frac{128}{\eta b^{-1}}$ of its neighbors. 
Hence, $d^+(u)$ can be increased by at most $\max \{
\left( d^+(u) \frac{\eta b^{-1} }{128} \right),  1 \}$, before re-setting the value $d^+_v(u)$. 
In particular, we have: $
d^+(u) < d^+_v(u) + \left( d^+_v(u) \frac{\eta b^{-1} }{128} \right) + 1$.

If $\theta = 1$, the statement is clear. Otherwise, we claim that $d^+(u) \geq \frac{b}{10}$ at all times. 
Observe first that if $b\leq 10$, this statement is certainly true, so assume otherwise:
denote by $x$ be the first non-isolated vertex, incident to at least one out-edge, such that $d^{+}(x) < \frac{b}{10}$. 
The first inserted edge in $G$ that caused $x$ to be non-isolated has to contribute at least $\frac{b}{10}$ out-edges to $x$. Thus, for this to occur, $x$'s out-degree must be decremented during a call to Algorithm~\ref{alg:deletion_alphaless}. 
If there still exists an out-edge $e_x$ of $x$, then $e_x$ must point to by a vertex of degree at least $\frac{9b}{10} \geq 2d^+(x)$, 
and so $x$ would have recursed on this neighbor instead of decrementing its own degree -- a contradiction.

We claim that $d^+(u) \geq \frac{b}{10}$ implies that $d_v^+(u) \geq \frac{b}{10}$. Indeed, if $d_v^+(u)$ were smaller than $\frac{b}{10}$ then because we update the out-neighbors of $d^+(u)$ so frequently: $d_v^+(u) = d^+(u)$ -- a contradiction. 
It follows from $\eta \geq 1280$ that $1 \leq d^+_v(u) \frac{\eta b^{-1} }{128}$. Thus, we recover the statement in the second case.
\end{proof}

\begin{lemma}
\label{lemma:decrement_bound}
Let for a vertex $v$, $d^+(v)$ be decremented by Algorithm~\ref{alg:deletion_alphaless}. 
Then  $\forall u \in N^-(v)$:
\[
d^+_v(u) \leq \left( 1 + \frac{3 \eta b^{-1}}{4} \right) d^+(v) + \theta 
\]
\end{lemma}

\begin{proof}
Per definition of Algorithm~\ref{alg:deletion_alphaless}, we know that
$d^+_v(x) \leq (1 + \frac{\eta b^{-1} }{2}) d^+(v) + \theta$ for: 
\newline \noindent $x \gets$ First(First(LinkedList of $N^-(u)$) (else, the algorithm flips $\overrightarrow{xv}$ instead of decrementing $d^+(v)$). 
Recall that the set $N^-(v)$ is stored in buckets $B_j$ with $j$ from high to low. 
Any $u \in N^-(v)$ is either in the same bucket as $x$ or a later bucket. Thus:
\begin{align*}
 d^+_v(u) & \left( 1 + \frac{\eta b^{-1}}{64} \right)^{-1} \leq d^+_v(x) \leq \left(1 + \frac{\eta b^{-1} }{2} \right) d^+(v) + \theta \Rightarrow   \\
d^+_v(u) &%\leq d^+_v(x) \leq
\left( 1 + \frac{\eta b^{-1}}{64} \right) \left(1 + \frac{\eta b^{-1} }{2} \right) d^+(v) + \theta \left( 1 + \frac{\eta b^{-1}}{64} \right)  \Rightarrow \\
d^+_v(u)  &\leq \left( 1 + \frac{\eta b^{-1}}{64} \right) \left(1 + \frac{\eta b^{-1} }{2} \right) d^+(v) + \frac{\eta b^{-1}}{64} d^+(v) + \theta
\end{align*}
Noting that $\eta b^{-1} < 1$ then implies the lemma.
\end{proof}

\begin{lemma}
\label{lemma:increment_bound}
Let for a vertex $u$, $d^+(u)$ be incremented by Algorithm~\ref{alg:deletion_alphaless}. 
Assuming that $\eta b^{-1} < 1$, then $\forall v \in N^+(u)$:
\[
d^+(u) \leq \left( 1 + \eta b^{-1} \right) d^+(v) + 2\theta 
\]
\end{lemma}

\begin{proof}
Consider the last time $t$ where $d^+(v)$ was decremented. At this point, we have by Lemma \ref{lemma:decrement_bound} that $d^+_v(u)_t \leq \left( 1 + \frac{3 \eta b^{-1}}{4} \right) d^+(v)_t + \theta $. 
Lemma \ref{lemma:offset_insertion} gives us that $d^+(u)_t < d^+_v(u)_t + \left( d^+_v(u)_t \frac{\eta b^{-1} }{64} \right) + \theta$. 
Finally, since at this point $d^+(u)$ can have increased by at most $d_{t}^+(u)_t\frac{\eta b^{-1}}{128}$, as new out-neighbours are immediately checked and put at the end of the queue. Combining these things, we find that:
\begin{align*}
d^+(u) &\leq (1+\frac{\eta b^{-1}}{128})d^+(u)_t \leq (1+\frac{\eta b^{-1}}{128})\cdot{}(1+\frac{\eta b^{-1} }{64}) d^+_v(u)_t \\
&\leq ((1 + 3\eta b^{-1} / 4) \cdot d^+(v)_t + \theta) (1+\frac{\eta b^{-1}}{128})\cdot{}(1+\frac{\eta b^{-1} }{64}) \\
& \leq (1 + \eta b^{-1}) \cdot d^+(v)_t + (1 + \eta b^{-1})\theta \leq (1 + \eta b^{-1}) \cdot d^+(v) + 2\theta
\end{align*}
Here, we have used the assumption that $\eta b^{-1} < 1$ multiple times. Note also that since $d^+(v)$ was last decremented at time $t$, it must be the case that $d^{+}(v)_t \leq d^{+}(v)$.
\end{proof}

\begin{lemma}[Maintaining Invariant~$\theta$]
\label{lemma:invariant_maintenance}
Whenever Algorithms~\ref{alg:insertion_alphaless} and~\ref{alg:deletion_alphaless} terminate, they maintain an orientation $\overrightarrow{G}^b$ where: for each edge $\overrightarrow{uv}$ in $\overrightarrow{G}^b$ it must be that: $
d^+(u) \leq (1 + \eta b^{-1} ) \cdot  d^+(v) + 2 \theta.$
\end{lemma}
%\todo{Here we are assuming that $\eta b^{-1} < 1$, right?}
\begin{proof}
Consider for the sake of contradiction the first time when our invariant would be violated.\footnote{Many lemmas in this section assume that $d^+(v) \geq 1$. Setting $b \geq 2$ ensures this.}
The violation cannot be caused by an increment due to Lemma~\ref{lemma:increment_bound}, so it must not be violated due to an insertion. Hence WLOG, we may assume that a deletion caused the violation
i.e. that our sequence of flips resulted in a graph $\overrightarrow{G}^b$ where the out-degree of a vertex $v$ was decreased by 1.
These deletions removed an edge $\overrightarrow{vw}$. Afterwards, there must exist some vertex $u$ where:
\begin{itemize}[noitemsep, nolistsep]
    \item The graph $\overrightarrow{G}$ contains the edge $\overrightarrow{uv}$, and  $d^+(u) > (1 + \eta b^{-1} ) \cdot  d^+(v) + 2 \cdot \theta$.
\end{itemize}

\noindent
Since this is our first violation, we know that: $
    d^+(u) \leq (1 + \eta b^{-1}) (d^+(v) + 1) + 2 \cdot \theta$. 
    
\noindent
Moreover, since we just decremented $d^+(v)$ we know by Lemma~\ref{lemma:decrement_bound} that: 

\begin{align*}
d^+_v(u) &\leq \left( 1 + \frac{3 \eta b^{-1}}{4} \right) d^+(v) + \theta  &\Rightarrow &\textnormal{\tiny by Lemma~\ref{lemma:offset_alltimes}} \\
%
   % d^+(u) &-  d^+(u) \cdot \frac{\eta b^{-1} }{64}  \leq \left(1 + \frac{3 \eta b^{-1} }{4} \right) d^+(v)  + \theta&\Rightarrow & \\
    %
    d^+(u) &\leq \left(1 + \frac{3 \eta  b^{-1}}{4} \right) d^+(v) +  d^+(u) \cdot \frac{\eta b^{-1} }{64} + \theta &\Rightarrow &\textnormal{\tiny because of first violation} \\
    d^+(u) &\leq \left(1 + \frac{3 \eta  b^{-1}}{4} \right) d^+(v)  + 
    \bigg( (1 + \eta b^{-1}) (d^+(v) + 1) + 2 \cdot \theta \bigg) 
    \cdot \frac{\eta b^{-1} }{64} + \theta   &\Rightarrow &\textnormal{\tiny by $(1 + \eta b^{-1}) < 2$} \\
    d^+(u) &\leq \left(1 + \frac{3 \eta  b^{-1}}{4} \right) d^+(v)+\bigg( 2 \cdot (d^+(v) + 1) + 2 \theta \bigg) \cdot \frac{\eta b^{-1} }{64} + \theta &\Rightarrow &\textnormal{\tiny observing that $d^+(v)$ is at least 1} \\
        d^+(u) &\leq \left(1 + \frac{3 \eta  b^{-1}}{4} \right) d^+(v) + 8  \cdot d^+(v)\cdot \frac{\eta b^{-1} }{64} + \theta < (1 + \eta  b^{-1}) d^+(v) + \theta & &
    \end{align*}
\noindent This contradicts our assumption that $d^+(u) > (1 + \eta b^{-1} ) \cdot  d^+(v) + 2 \cdot \theta$.
\end{proof}

\subsection{Adaptive out-orientations}

We showed that whenever Algorithms~\ref{alg:insertion_alphaless} and~\ref{alg:deletion_alphaless} terminate, they maintain a graph $\overrightarrow{G}^b$ where Invariant~$\theta$ holds. 
What remains, is to show how much time they require to terminate and to show how this allows us to maintain an approximate minimal out-orientation. 
Just as in the previous section, we restrict our proofs to be insertion only (as the proof for deletions is symmetrical).
To this end, we observe the following consequence of Lemma~\ref{lemma:offset_insertion} (for whenever an insertion invalidates $\overrightarrow{ab}$):

\begin{observation}
\label{obs:critical_ineq}
Suppose the update algorithm reorients an edge $x \rightarrow y$ in $\overrightarrow{G}^b$. Then we have that $d^{+}(x) > (1 + \eta b^{-1} / 4)( d^+(y) -1) +  \theta(1-\frac{\eta b^{-1} }{128})$.
\end{observation}
\begin{proof}
If the reorientation happens during an insertion, we have $d^+(x) + 1 > (1 + \eta b^{-1} / 2) d^+(y) +  \theta$. 
If it is during a deletion, we have: $d_y^+(x) > (1 + \eta b^{-1} / 2)( d^+(y) -1) +  \theta$. 
Due to the Round-Robin scheme, $d^{+}(x)$ can be decremented by at most $
\left( d^+(x) \frac{\eta b^{-1} }{128} \right)$ before it re-sets the value $d^+_y(x)$. Note that if $d^+(x) \frac{\eta b^{-1} }{128} < 1$, $d^+(x)$ is of course decremented by 1, but then the estimate $d^+_y(x)$ is update immediately and therefore it is precise. 
It follows that $d^+(x) \geq (1-\frac{\eta b^{-1} }{128}) d^{+}_y(x)$ at all times, when the algorithm is considering $y$ during a deletion, and so we have:
$d^+(x) \geq (1-\frac{\eta b^{-1} }{128}) ((1 + \eta b^{-1} / 2)( d^+(y) -1) +  \theta) \geq (1 + \eta b^{-1} / 4)( d^+(y) -1) +  \theta(1-\frac{\eta b^{-1} }{128})$. Note that here we are relying on the fact that $d^+(y) -1$ is non-negative.
\end{proof}
\noindent
We briefly remark that if at all times $\frac{\eta b^{-1}}{4}d^{+}(y) + \theta(1-\frac{\eta b^{-1} }{128})\geq 1 + \eta b^{-1} / 4,$ then $d^{+}(x) > d^{+}(y)$ in the observation above. 
This ensures that there always exists a way to direct an edge $(u, v)$, that does not violate Invariant~$\theta$. 
We show how to assure this, depending on the invariant that we wish to maintain:

\subparagraph{Invariant 1.}
For $\theta = 1$, ensuring that $\frac{\eta b^{-1}}{4}d^{+}(y) \geq \eta b^{-1} / 4 + \eta b^{-1} / 128$ is sufficient, and so ensuring that $d^{+}(y) \geq 2$ is enough. Note that if $\eta b^{-1} < 1$, then it is sufficient to pick $b = 6$  in order to ensure this. 

\subparagraph{Invariant 0.}
On the other hand for $\theta = 0$, we note that if $\eta b^{-1} < 1$ and $\eta \geq 1280$, then the arguments from Lemma~\ref{lemma:offset_insertion} show that $d^{+}(y) \geq \frac{b}{10}$ at all times, and so we know that $\frac{\eta b^{-1}}{4}d^{+}(y) \geq 3$, and so certainly this inequality still holds. 
We finally note that for all non-negative values of $\theta$, we have that if $d^{+}(y) \geq \frac{128}{\eta b^{-1}}$, then $d^{+}(y) - 1 \geq d^{+}(y)(1-\frac{\eta b^{-1} }{128})$, and so in this case the observation shows that in fact $d^{+}(x) > (1 + \eta b^{-1} / 4)(1-\frac{\eta b^{-1} }{128})d^{+}(y) +  \theta(1-\frac{\eta b^{-1} }{128}) \geq (1 + \eta b^{-1} / 6)d^{+}(y)$ for all non-negative choices of $\theta$.

These observations are  intuitively why we maintain Invariant~\ref{inv:degrees_additive} and Invariant~\ref{inv:degrees} in $O(\log^2 n \log \rho)$ and $O(\log^3 n \log \rho)$ time respectively using perceived values and a round-robin strategy:

\subsubsection{Efficiently maintaining Invariant~\ref{inv:degrees_additive}.}

\begin{theorem}
\label{thm:fast_butworse}
Let $G$ be a graph subject to edge insertions and deletions. 
Choose $c = 2, b = 6$,  $\gamma$ constant and $\eta$ be such that $\eta b^{-1} < \frac{1}{\log n \max \{  \gamma^{-1}, 1\}}$.
We can choose $\gamma$ to maintain an out-orientation $\overrightarrow{G}^b$ and $\overrightarrow{G}$ in $O( \log^2 n \log \rho)$ time s.t.:
\begin{itemize}[noitemsep, nolistsep]
    \item Invariant~\ref{inv:degrees_additive} holds for $\overrightarrow{G}^b$, and
    \item $\forall u$ in $\overrightarrow{G}$ and $\overrightarrow{G}^b$, the out-degree is at most $O(\rho + \log n)$,
\end{itemize}
where $\rho$ is the density of $G$ at the time of the update.
\end{theorem}

\begin{proof}
The proof is near-identical to that of Theorem~\ref{thm:invariant2}. 
Maintaining Invariant~\ref{inv:degrees_additive} for $\overrightarrow{G}^b$ (for $b \geq 2$) guarantees that $\forall u$ in $\overrightarrow{G}$ and $\overrightarrow{G}^b$, the out-degree is at most $O(\rho + \log n)$.
What remains is to show that our algorithms are correct and terminate in $O(\log^2 n \log \rho)$ time:
%(in a similar fashion like Theorem~\ref{thm:invariant2}).

\noindent
\subparagraph{Correctness.}
We may orient any edge $(u, v)$ in $\overrightarrow{G}^b$: suppose that we cannot insert $\overrightarrow{vu}$.
Then per definition: $d^+(v) + 1 > (1 + \eta b^{-1} / 2) d^+_v(u) + 1$. 
By Observation~\ref{obs:critical_ineq}, this implies that $d^+(u) < d^+(v)$. 
Now suppose that the perceived value $d^+_u(v)$ is larger than the actual value $d^+(v)$. 
Then $d^+(u) < d^+_u(v) \leq (1 + \eta b^{-1} / 2) d^+_u(v)$ and our algorithm will insert $\overrightarrow{uv}$.
If the perceived value is smaller, then $d^+(v)$ must have been incremented at some point.
This fact allows us to apply Lemma~\ref{lemma:offset_insertion} to note that $d^+(u) < (1 + \eta b^{-1} /64 ) d^+_u(v) + 1 < (1 + \eta b^{-1} /2 ) d^+_u(v) + 1$.
Thus $d^+(u) + 1 \leq 1 + \eta b^{-1} /2 ) d^+_u(v) + 1$ and our algorithm will insert $\overrightarrow{uv}$ instead. 
Inserting $\overrightarrow{uv}$ may imply that for some edge $\overrightarrow{ux}$: $d^+(u) + 1 > (1 + \eta b^{-1} / 2) d^+_u(x) + 1$.
We flip $\overrightarrow{ux}$ and by Observation~\ref{obs:critical_ineq} and the subsequent remark: $d^+(u) > d^+(x)$. So, this directed chain of flipped edges has finite length.

\subparagraph{Recursions.}
Suppose that we insert an edge $\overrightarrow{uv}$ and that we recursively flip a directed chain of edges of length $T$ to reach a vertex $x_T$. We now make a case distinction: 

Let $\rho < \frac{128}{\eta b^{-1}}$. The out-degree $d^+(u)$ is at most $O(\frac{1}{\eta b^{-1}})$ and for any two consecutive vertices $(y, z)$ on the chain: $d^+(z) < d^+(y)$. Thus $T \in O(\frac{1}{\eta b^{-1}}) = O(\log n)$. 

Let $\rho \geq \frac{128}{\eta b^{-1}}$. Our analysis shows two things: the out-degree $d^+(u)$ is at most $O(\rho)$ and for any two consecutive vertices $(y, z)$ on the chain: $d^+(y) > (1 + \eta b^{-1} / 6) d^+(z)$ as long as both have out-degree at least $\frac{128}{\eta b^{-1}}$.
As soon as this is no longer the case, we are in the above case and the rest of the chain has length at most $O(\frac{1}{\eta b^{-1}}) = O(\log n)$.
It follows that $d^+(u) > (1 + \eta b^{-1} / 6)^T d^+(x_T)$ and so $T \in O(\log n \log \rho)$.

\subparagraph{Time spent.}
Each insertion in $G$ triggers $O(b) = O(1)$ insertions in $\overrightarrow{G}^b$.
For each insertion in $\overrightarrow{G}^b$, our algorithm spends $O(1)$ time per function call, recurses $O(\log n \log \rho)$ times. We inform or look at $O( (\eta b^{-1})^{-1} ) = O(\log n)$ out-neighbors of every vertex in the chain (in round robin). 
Whenever we update for an out-neighbor $y \in N^+(x_T)$ the value $d_y^+(x_T) \gets d^+(x_T)$ we may need to update the data structure that $y$ has on the set $N^-(y)$. 
By Lemma~\ref{lemma:increment_bound}, the new value of $d^+_v(x)$ can be increased by at most a factor $(1 + \frac{\eta b^{-1}}{64})$ which ensures that we need to move $x_T$ by at most $O(1)$ buckets in $N^-(y)$.
Thus, our algorithm spends $O(\log^2 n \log \rho)$ time per insertion in $G$. Deletions are a factor $O(\log n)$ faster, because they only do the round-robin when the out-degree definitively changes.
\end{proof}

\subsubsection{Efficiently maintaining Invariant~\ref{inv:degrees}.}

\begin{theorem}
\label{thm:fastnonadd}
Let $G$ be a graph subject to edge insertions and deletions. 
Choose $\eta > 1280$, $\gamma$ constant, and $b$ such that $\eta b^{-1} \leq  \frac{1}{100 \log n \max \{ \gamma^{-1}, 1 \} }$. 
Let $G^b$ be the graph where each edge in $G$ is duplicated $b$ times. 
Our algorithms maintain out-orientations $\overrightarrow{G}^b$ and $\overrightarrow{G}$ s.t.:
\begin{itemize}[noitemsep, nolistsep]
    \item Invariant~\ref{inv:degrees} holds for $\overrightarrow{G}^b$,
    \item $\forall u$, the out-degree $d^+(v)$ in $\overrightarrow{G}^b$ is at most $O(b \cdot \rho)$, and
    \item $\forall u$, the out-degree of $u$ in $\overrightarrow{G}$ is at most $O(\rho)$.
\end{itemize}
We use $O( b^3 \cdot \log \rho) = O( \log^3 n \log \rho)$ time per operation in the original graph $G$. Here, $\rho$ is the density of $G$ at the time of the update. 
\end{theorem}
\begin{proof}
The proof is near-identical to the proof of Theorem~\ref{thm:invariant2}.
Maintaining Invariant~\ref{inv:degrees} for $\overrightarrow{G}^b$ gives the desired upper bound on the out-degrees of vertices in $\overrightarrow{G}^b$ and $\overrightarrow{G}$. What remains is to show that our algorithms terminate in the appropriate time. 

\subparagraph{Correctness.}
We claim that for all vertices $u$ in $\overrightarrow{G}^b$ which have at least one edge $\overrightarrow{uv}$: $d^+(u)$ and $d^+_v(u)$ are at least $4 \left(\eta b^{-1}  \right)^{-1}$.
Indeed, if $d^+(u)$ is incident to at least one edge then its out-orientation is at least $\frac{b}{2}$ and $\eta > 1280$. 
Suppose that we cannot insert $\overrightarrow{vu}$. Then $d^+(v) + 1 > (1 + \eta b^{-1} / 2) d^+_v(u)$. By Observation~\ref{obs:critical_ineq} this implies that $d^+(v) + 1 > (1 + \eta b^{-1} / 4) d^+(y)$.
We may again apply Lemma~\ref{lemma:increment_bound} to conclude that we may insert $\overrightarrow{uv}$ instead.

Just as in Theorem~\ref{thm:fast_butworse}, inserting $\overrightarrow{uv}$ may cause our algorithm to flip a directed chain of edges. Consider an edge $\overrightarrow{yz}$ in this chain. 
Before inserting $\overrightarrow{uv}$ it must be that $d^+(y) + 1 > (1 + \eta b^{-1} / 2) d^+_y(z)$ and, by Observation~\ref{obs:critical_ineq}, $d^+(y) > d^+(z)$. This chain must hence be finite. 

\subparagraph{Recursions.}
Suppose that we insert an edge $\overrightarrow{uv}$ and that we recursively flip a directed chain of edges of length $T$ to reach a vertex $x_T$. Our analysis shows two things: the out-degree $d^+(u)$ is at most $O(\rho)$ and for any two consecutive vertices $(y, z)$ on the chain: $d^+(y) > (1 + \eta b^{-1} / 6) d^+(z)$, since $d^{+}(y) \geq \frac{b}{10} \geq \frac{128}{\eta b^{-1}}$ as we discussed earlier. 
It follows that $d^+(u) > (1 + \eta b^{-1} / 4)^T d^+(x_T)$ and so $T \in O( b \log \rho) = O(\log n \log \rho)$.

\subparagraph{Time spent.}
Each insertion in $G$ triggers $O(b)$ insertions in $\overrightarrow{G}^b$.
For each insertion in $\overrightarrow{G}^b$, our algorithm  recurses $O(b \log \rho)$ times. At each step of the recursion look at or inform at most $O( (\eta b^{-1})^{-1} ) = O(b)$ out-neighbors in a round-robin fashion.
Whenever we update for such an out-neighbor $y \in N^+(x_T)$ the value $d_y^+(x_T) \gets d^+(x_T)$ we may need to update the data structure that $y$ has on the set $N^-(y)$. 
By Lemma~\ref{lemma:increment_bound}, the new value of $d^+_v(x)$ can be increased by at most a factor $(1 + \frac{\eta b^{-1}}{4})$ which ensures that we need to move $x_T$ by at most $O(1)$ buckets in $N^-(y)$.
Thus, our algorithm spends $O(b \cdot b \cdot b \log \rho) = O( \log^3 n \log \rho)$ time per insertion in $G$.
Again, deletions require a factor $O(\log n)$ fewer time: as the round robin is executed only for the last step in the recursion.
\end{proof}

\section{Obtaining  \texorpdfstring{$(1 + \eps)$}{(1+e)} Approximations }
\label{sec:onepluseps}
Finally, we note that we can choose our variables carefully to obtain a $(1 + \eps)$ approximations  of the maximal subgraph density or minimal out-degree. 
 Theorem~\ref{thm:fastnonadd} implies that, for suitable choices of $\eta$ and $b$, we can for any graph $G$ maintain a directed graph $\overrightarrow{G}^b$ (where $G^b$ is the graph $G$ with every edge duplicated $b$ times) such that $\overrightarrow{G}^b$  maintains Invariant~\ref{inv:degrees}.
 By Theorem~\ref{thm:structural}, $\overrightarrow{G}^b$ approximates the densest subgraph of $G$ and the minimal out-orientation of $G^b$ (where the approximation factor is dependent on $\beta$ and $\eta$). 
  The running time of the algorithm is $O( b^3 \cdot \log \alpha)$ where $\alpha$ is the arboricity of the graph. 
  In this section we show that for any $0 < \eps < 1$, whenever we choose $b \in O(\eps^{-2} \log n)$ to ensure that $\overrightarrow{G}^b$ maintains a:
  
  \begin{itemize}
      \item $(1 + \eps)$-approx. of \textbf{the maximal densest subgraph} of $G$ in $O(\eps^{-6} \log^3 n \log \alpha)$ time.
      \item $(1 + \eps)$-approx. of the minimal out-orientation of $G^b$. This implies an explicit \textbf{$(2 + \eps)$-approximation of the minimal out-orientation} of $G$  in $O(\eps^{-6} \log^3 n \log \alpha)$ time.
      \item  $(1 + \eps)$-approx. of the minimal out-orientation of $G^b$. 
      Through applying clever rounding introduced by Christiansen and Rotenberg~\cite{christiansenICALP} we obtain an explicit \textbf{$(1 + \eps)$-approximation of the minimal out-orientation} of $G$. 
   By slightly opening their black-box algorithm, we can show that applying their technique does not increase our running time.  Thus, our total running time is thus $O(\eps^{-6} \log^3 n \log \alpha)$.
  \end{itemize}

\subsection*{Obtaining a $(1 + \eps)$ Approximation for Densest Subgraph}

\begin{corollary}\label{cor:eps_subgraphdensity}
Let $G$ be a dynamic graph subject to edge insertions and deletions with  adaptive maximum subgraph density $\rho$. Let $G^b$ be $G$ where every edge is duplicated $b$ times. Let $0 \leq \epsilon < 1$. We can maintain an orientation $\overrightarrow{G}^b$ such that 
\[
\rho \leq  b^{-1} \cdot \Delta( \overrightarrow{G}^b) \leq (1+\eps)\rho
\]
with update time $O(\eps^{-6}\log^3(n)\log \rho )$ per operation in $G$.
\end{corollary}
\begin{proof}
We apply Theorem~\ref{thm:fastnonadd} in order to maintain an out-orientation with $\rho(G') \leq \Delta(\overrightarrow{G}^b) \leq (1+\gamma)(1+\eta\cdot{}b^{-1})^{k_{\max}} \rho(G')$. 
By setting $\eta > 1280$, $b = O(\eps'^{-2}\eta \log{n})$ and $\gamma = \eps'$, we satisfy the conditions of the Theorem. Since $k_{\max} \leq \log_{1+\gamma} n = O(\gamma^{-1}\log n)$, we find that
\[
(1+\eta\cdot{}b^{-1})^{k_{\max}} \leq e^{\eta b^{-1} \cdot k_{\max}} \leq e^{\eps'} \leq 1+2\eps'
\]
where the last inequality comes from the fact that for $0 \leq x \leq 1$, we have $e^x \leq 1+2x$. Finally, by setting $\eps' \leq \frac{\eps}{10}$, we recover the statement.
\end{proof}

\begin{observation}
The algorithm of Corollary~\ref{cor:eps_subgraphdensity} can in $O(1)$ time per operation, maintain the integers: $b^{-1}$, $\Delta(\overrightarrow{G}^b)$ and thus a $(1 + \eps)$ approximation of the value of the density of $G$. 
\end{observation}

\noindent
However, to actually output any such realizing subgraph, a bit more of a data structure is needed:

\begin{lemma} \label{lemma:SDE}
For a fully-dynamic graph $G$, there is an algorithm 
that explicitly maintains a $(1+\eps)$ approximation of the maximum subgraph density in 
%$O(\eps^{-1} \log^2 n)$ additional 
$O(\eps^{-6}\log ^3 n \log \alpha )$ total
time per operation, and that
can output a subgraph realizing this density in  $O(c)$ time where $c$ is the size of the output.
% and the size of the densest subgraph in $O(1)$ time.
\end{lemma}

%\begin{lemma} \label{lemma:SDE}
%Consider an algorithm that maintains for a dynamic graph $G$, a directed graph $\overrightarrow{G}^b$ for which Invariant~$0$ holds (with $\eta > 1280$, $b = O(\eps'^{-2}\eta \log{n})$ and $\gamma = \eps'$). Moreover, let for every operation in $G$, this algorithm change the degree of $r_o$ vertices in $\overrightarrow{G}^b$ in $t$ total time. 
%Then  we can dynamically maintain a data structure in $O(t + r_o \log n)$ time that can report a $(1+\eps)$ approximation of the densest subgraph of $G$ in $O( \eps^{-1} \log^2 n + c)$ time where $c$ is the size of the output.
%\end{lemma}

\begin{proof}
We use Corollary~\ref{cor:eps_subgraphdensity} to dynamically maintain an orientation $\overrightarrow{G}^b$ in $O(\eps^{-6}\log^3(n)\log \rho )$ per operation in $G$.
Recall (Theorem~\ref{thm:structural}) that we defined for non-negative integers $i$ the sets:
\[
T_i := \left\{ v \mid d^+(v) \geq \Delta\left(\overrightarrow{G}^b\right) \cdot \left(1 + \eta\cdot b^{-1}\right)^{-i} \right\} 
\]

(note that since we maintain Invariant~\ref{inv:degrees}, the constant $c$ in the previous definition is zero).

Let $k$ be the smallest integer such that $|T_{k+1}| < (1 + \gamma) |T_k|$). Moreover, we showed in Corollary~\ref{cor:eps_subgraphdensity} that $k$ is upper bounded by $O( \eps^{-1} \log n)$.  
Salwani and Wang~\cite{sawlani2020near} show that (the induced subgraph of the vertex set) $T_k$ is an approximation of the densest subgraph of $\overrightarrow{G}^b$ (and therefore of $G$).
We store the vertices of $\overrightarrow{G}^b$ as leaves a balanced binary tree, sorted on their outdegree. 
Since every change in $G$, changes at most $O(b) = O(\eps^{-2} \log n)$ out-degrees in $\overrightarrow{G}$, we can maintain this binary tree in $O(\eps^{-2} \log^2 n)$ additional time per operation in $G$. 

Each internal node of the balanced binary tree stores the size of the subtree rooted at that node. 
Moreover, we store the maximal out-degree $\Delta(\overrightarrow{G^b})$ as a separate integer, and a doubly linked list amongst the leaves.

After each operation in $G$, for each integer $i \in [0, \eps^{-1} \log n]$, we determine how many elements there are in $T_i$ as follows: 
first, we compute the value $V_i = \Delta(\overrightarrow{G}^b) \cdot (1 + \eta\cdot b^{-1})^{-i} - c\sum_{j = 1}^i (1 + \eta\cdot b^{-1})^{-j}$ in $O(1)$ time (using the previous sum $c\sum_{j = 1}^{i-1} (1 + \eta\cdot b^{-1})^{-j}$ from when we computed $V_{i-1}$). 
Then, we identify in $O(\log n)$ time how many vertices have out-degree at least $V_i$ (thus, we determine the size of $T_i$). 
It follows that we identify $T_k$ in $O(\eps^{-1} \log^2 n)$ additional time. We store a pointer to the first leaf that is in $T_k$. 
If we subsequently want to output the densest subgraph of $G$, we traverse the $c$ elements of $T_k$ in $O(c)$ total time by traversing the doubly linked list of our leaves.
\end{proof}

\subparagraph{Related Work}
While results for densest subgraph \cite{BahmaniKV12,BhattacharyaHNT15,EpastoLS15} can be used to estimate maximum degree of the best possible out orientation, it is also interesting in its own right.
Sawlani and Wang~\cite{sawlani2020near} maintain a $(1 - \eps)$-approximate densest subgraph in worst-case time $O(\eps^{-6}\log^4 n )$ per update where they maintain an \emph{implicit} representation of the approximately-densest subgraph. They write that they can, in $O(\log n)$ time, identify the subset $S \subseteq V$ where $G[S]$ is the approximately-densest subgraph and they can report it in $O(|S|)$

\subsection*{Obtaining an almost $(1 + \eps)$ Approximation for Minimal Out-orientation }

By Corollary~\ref{cor:eps_subgraphdensity}, we can dynamically maintain for every graph $G$,  a directed graph $\overrightarrow{G}^b$ (where each edge in $G$ is duplicated $b$ times) such that the maximal out-degree in $\overrightarrow{G}^b$ is at most a factor $(1 + \eps)$ larger than the minimal out-orientation of $G^b$. 
For every edge $(u, v)$ in $G$, we can now store a counter indicating how many edges point (in $G^b$) from $u$ to $v$, or the other way around. 
The naive rounding scheme, states that the edge $(u, v)$ is directed as $\overrightarrow{uv}$ whenever there are more edges directed from $u$ to $v$. For any edge, we can decide its rounding in $O(1)$ time, thus we conclude:

\begin{observation}\label{obs:rounding}
We can maintain for a graph $G$ an orientation $\overrightarrow{G}$ where each vertex has an out-degree of at most 
$(2+\varepsilon)\alpha + 1$ 
 with update time $O(\eps^{-6}\log^3(n)\log \rho )$ per operation.
\end{observation}

\noindent
Obtaining a $(1 + \eps)$-approximation of the minimal out-orientation of $G$ is somewhat more work. 
Christiansen and Rotenberg~\cite{christiansenICALP} show how to dynamically maintain an explicit out-orientation on $G$ of at most $(1 + \eps) \alpha + 2$ out-edges. In their proofs, Christiansen and Rotenberg~\cite{christiansenICALP} rely upon the algorithm by Kopelowitz, Krauthgamer, Porat and Solomon~\cite{KopelowitzKPS13}. 
By replacing the KKRS algorithm by ours in a black-box like manner, we obtain the following:

\begin{theorem}
\label{thm:epsapprox}
Let $G$ be a dynamic graph subject to edge insertions and deletions. We can maintain an orientation $\overrightarrow{G}$ where each vertex has an out-degree of at most $(1+\varepsilon)\alpha + 2$ 
 with update time $O(\eps^{-6} \log^3 n \log \alpha)$ per operation in $G$, where $\alpha$ is the arboricity at the time of the update.
\end{theorem}

\noindent
The proof follows immediately from the proof Theorem 26 by Christiansen and Rotenberg~\cite{christiansenICALP} (using Corollary~\ref{cor:eps_subgraphdensity} as opposed to~\cite{KopelowitzKPS13}). 
For the reader's convenience, we will briefly elaborate on how this result is obtained and how we can apply Corollary~\ref{cor:eps_subgraphdensity}.
For the full technical details, we refer to the proof of Theorem 26 in~\cite{christiansenICALP}.

\begin{enumerate}
    \item 
Christiansen and Rotenberg consider a graph $G$ with arboricity $\alpha$.
Moreover, they construct a directed graph $\overrightarrow{G}^b$ which is the graph $G$ where every edge in $G$ is duplicated $b \in O(\eps^{-2} \log n)$ times.\footnote{In \cite{christiansenICALP}, Christiansen and Rotenberg choose the duplication constant to be $\gamma$ and write $G^\gamma$.} 
Every operation in $G$ triggers $O(b)$ operations in $\overrightarrow{G}^b$.

\item On the graph $G^b$, they run the algorithm by~\cite{KopelowitzKPS13} to maintain an orientation of $\overrightarrow{G}^b$ where each vertex has an out-degree of at most $\Delta(\overrightarrow{G}^b) = (1 + \eps)\alpha \cdot b + \log_{(1 + \eps)} n$. 
The KKPS algorithm uses per operation in $G^b$:\footnote{Christiansen and Rotenberg deliberately use the adaptive variant of KKPS.}
\begin{itemize}[noitemsep, nolistsep]
    \item $O\left( \left(\Delta(\overrightarrow{G}^b) \right)^2 \right) = O \left( (1 + \eps)^2 \alpha^2 b^2 + \eps^{-4} \log^2 n \right)  = O( \eps^{-4} \alpha^2 \log^2 n)$ time, and
    \item $O\left( \Delta(\overrightarrow{G}^b) \right) = O \left( (1 + \eps)\alpha b + \eps^{-2} \log n \right)  = O( \eps^{-2} \alpha \log n)$ \emph{combinatorial changes } in $\overrightarrow{G}^b$.
    (here, a combinatorial change either adds, removes, or flips an edge in $\overrightarrow{G}^b$).
\end{itemize}

\item Finally, they deploy a clever rounding scheme to transform the orientation $\overrightarrow{G}^b$ into an orientation of $G$ where the out-degree of each vertex in $\overrightarrow{G}$ is at most a factor $\frac{1}{b}$ the out-orientation of $\overrightarrow{G}^b$, plus two.
Thus, they ensure that each vertex has an out-degree of at most: 
\[
(1 + \eps) \alpha + b^{-1} \log_{1 + \eps} n  + 2 \leq (1 + \eps) \alpha + \frac{\eps^2}{\log n} \cdot \frac{\log n}{\eps^2} + 2 = (1 + \eps) \alpha + 2
\]
They achieve this in $O(\log n)$ additional time per combinatorial change in $\overrightarrow{G}^b$.
Specifically:
\begin{itemize}
    \item They consider for every edge $(u, v)$ in $G$ its partial orientation (i.e. how many edges in $G^b$ point from $u$ to $v$ or vice versa). 
If the partial orientation contains sufficiently many edges directed from $u$ to $v$, the edge in $G$ gets rounded (directed from $u$ to $v$). 
\item Let $H$ be a (not necessarily maximal) set of  edges in $G$ whose direction can be determined in this fashion. They call $H$ a \emph{refinement}. Christiansen and Rotenberg choose $H$ such that in the directed graph $H$ each vertex has an out-degree of at most $(1 + \eps) \alpha$.
\item 
Christiansen and Rotenberg show that  $G - H$ is a forest.
For all edges in $G - H$, they no longer explicitly store the $b$ copies in $\overrightarrow{G}^b$. Instead, they store for edges in $G - H$ their (partial) orientation as an integer in $[0, b]$. 
The forest $G - H$ gets stored in a top tree where each interior node stores the minimal and maximal partial orientation of all its children.
For any path or cycle in $G - H$, they can increment or decrement all orientation integers by 1 in $O(\log n)$ time by lazily updating these maxima and minima in the top tree. 
For each edge in $G - H$, one can obtain the exact partial orientation in $O(\log n)$ additional time by adding all lazy updates in the root-to-leaf path of the top tree.
\item 
In addition, they show how to dynamically maintain a $2$-orientation on the forest $G - H$ in $O(\log n)$ update time per insertion in the forest. Adding the directed edges from the forest to $G$ ensures that each vertex has an out-degree of at most $(1 + \eps) \alpha + 2$.
\item For each combinatorial change in $\overrightarrow{G}^b$, they spend $O(\log n)$ time.
Specifically:
\begin{itemize}
    \item each combinatorial change in $\overrightarrow{G}^b$ may remove an edge from the forest. The edge can be rounded in $O(1)$ time and removed from the top tree in $O(\log n)$ time. 
    \item each combinatorial change may force an edge in the refinement $H$ out of the refinement, and into the forest (creating a cycle).
    \item When creating a cycle, the authors augment the cycle such that at least one edge may be added to the refinement. They (implicitly) increment or decrement all orientation integers along the cycle using the lazy top tree in $O(\log n)$ total time. 
    \item Augmenting a cycle causes the out-degree to remain the same for all elements on the cycle. Hence, the Invariants of KKPS (and our Invariant~\ref{inv:degrees}) to stay unchanged and the augmentation does not trigger any further operations in $\overrightarrow{G}^b$.
\end{itemize}
\item The final edge along the augmented path may subsequently be rounded and added to $H$. Thus, spending $O(\log n)$ time per combinatorial change in $\overrightarrow{G}^b$.
\end{itemize}
\end{enumerate}

\noindent
It follows through these three steps that the algorithm in \cite{christiansenICALP} has a running time of: 
\[
O\left( b \cdot \left( \left( \Delta(\overrightarrow{G}^b) \right)^2 +  \Delta(\overrightarrow{G}^b) \log n \right) \right) = O\left( \eps^{-6} \alpha^2 \log^3 n \right)
\]

\noindent
Given the results in this paper, we can instead apply our results as follows: 

\begin{enumerate}
    \item We again choose $b \in O(\eps^{-2} \log n)$. Each operation in $G$ triggers $O(b)$ operations in $\overrightarrow{G}^b$. 
    \item We apply Theorem~\ref{thm:fastnonadd} (or conversely Corollary~\ref{cor:eps_subgraphdensity}) to maintain $\overrightarrow{G}^b$ such that each vertex has an out-degree of at most $\Delta(\overrightarrow{G}^b) = (1 + \eps) \alpha b$. 
    We proved that this algorithm takes: 
    \begin{itemize}
        \item  $O(b^2 \log \alpha)$ time per operation in $\overrightarrow{G}^b$, but
        \item only triggers $O(b \log \alpha)$ combinatorial changes (edge flips) in $\overrightarrow{G}^b$.
    \end{itemize}
    \item Finally, we apply the rounding scheme by Christiansen and Rotenberg which requires $O(\log n)$ time per \emph{combinatorial change} in $\overrightarrow{G}^b$.
\end{enumerate}

\noindent
Our total running time is (our algorithm + rounding scheme per combinatorial change):
\[
O\left(b \cdot b \cdot b \log \alpha  + b \cdot b \log \alpha \cdot \log n \right) = O\left(\eps^{-6} \log^3 n \log \alpha\right).
\]

\subparagraph{Related Work}
Historically, four criteria are considered when designing dynamic out-orientation algorithms: the maximum out-degree, the update time (or the recourse), amortized versus worst-case updates, and the adaptability of the algorithm to the current arboricity. 

Brodal and Fagerberg~\cite{Brodal99dynamicrepresentations} were the first to consider the out-orientation problem in a dynamic setting. 
They showed how to maintain an $\mathcal{O}(\alpha_{\max})$ out-orientation with an amortized update time of $\mathcal{O}(\alpha_{\max}+ \log{n})$, where $\alpha_{\max}$ is the maximum arboricity throughout the entire update sequence.
Thus, their result is adaptive to the current arboricity as long as it only increases. 
He, Tang, and Zeh~\cite{HeTZ14} and Kowalik~\cite{10.1016/j.ipl.2006.12.006} provided different analyses of Brodal and Fagerbergs algorithm resulting in faster update times at the cost of worse bounds on the maximum out-degree of the orientations. 
Henzinger, Neumann, and Wiese~\cite{henzinger2020explicit} gave an algorithm able to adapt to the current arboricity of the graph, achieving an out-degree of $\mathcal{O}(\alpha)$ and an amortized update time \emph{independent} of $\alpha$, namely $O(\log^2 n)$.
Kopelowitz, Krauthgamer, Porat, and Solomon~\cite{KopelowitzKPS13} showed how to maintain an $\mathcal{O}(\alpha+\log n)$ out-orientation with a worst-case update time of $\mathcal{O}(\alpha^2 + \log^2 n)$ fully adaptive to the arboricity. 
Christiansen and Rotenberg~\cite{christiansenICALP,christiansenMFCS} lowered the maximum out-degree to $(1+\varepsilon)\alpha+2$ incurring a worse update time of $\mathcal{O}(\varepsilon^{-6}\alpha^2\log^3 n)$.
Finally, Brodal and Berglin~\cite{berglinetal:LIPIcs:2017:8263} gave an algorithm with a different trade-off; they show how to maintain an $\mathcal{O}(\alpha_{\max}+\log n)$ out-orientation with a worst-case update time of $\mathcal{O}(\log n)$. This update time is faster and independent of $\alpha$, however the maximum out-degree does not adapt to the current value of $\alpha$.

\section{Applications}
\label{app:applications}
In this section, we show how to combine our two trade-offs for out-orientations (theorems \ref{thm:fast_butworse}, \ref{thm:fastnonadd}) with existing or folklore reductions, obtaining improved algorithms for maximal matching, arboricity decomposition, and matrix-vector product.

\subsection{Maximal matchings}
  
  For our application in maximal matchings, we first revisit the following result.
  The authors have not seen this theorem stated in this exact generality in the literature, but similar statements appear in~\cite{PelegS16}, \cite{NeimanS16}, and \cite{berglinetal:LIPIcs:2017:8263}
  
\begin{lemma}[Folklore] \label{thm:folklore}
Suppose one can maintain an edge-orientation of a dynamic graph, 
that has $t_u$ update time, 
that for each update performs at most $r_u$ edge re-orientations (direction changes), and that maintains a maximal out-degree of $\le n_o$. Then there is a dynamic maximal matching algorithm\footnote{When the update time $t_u$ is worst-case, the number of re-orientations $r_u$ is upper bounded by $t_u$.} whose update time is $O(t_u+r_u+n_o)$.
\end{lemma}

\begin{proof}[Proof of Lemma~\ref{thm:folklore}]
Each vertex maintains two doubly-linked lists over its in-neighbors (one for the matched, and one for the available in-neighbors) called \emph{in-lists}
and a doubly-linked list of its out-neighbors called the \emph{out-list}. When a vertex becomes available because of an edge deletion, it may match with the first available in-vertex if one exists.
If no such in-vertex exists, it may propose a matching to its $\le n_o$ out-neighbors in the out-list, and then match with an arbitrary one of these if any is available. When a vertex $v$ changes status between matched and available, it notifies all vertices in its out-list, who move $v$ between in-lists in $O(1)$ time. Finally, when an edge changes direction, each endpoint needs to move the other endpoint between in- and out-lists. 

The bookkeeping of moving vertices between unordered lists takes constant time.
For each edge insertion or deletion, we may spend additionally $O(n_o)$ time  proposing to or notifying to out-neighbors to a vertex, for at most two vertices for each deletion or insertion respectively.
\end{proof}

\noindent
With this application in mind, some desirable features of out-orientation algorithms become evident: 
\begin{itemize}[noitemsep, nolistsep]
    \item we want the number of out-edges $\mout$ to be (asymptotically) low, and
    \item we want the update time to be efficient, preferably deterministic and worst-case.
\end{itemize}

\noindent
Here, a parameter for having the number of out-edges asymptotically as low as possible, can be sparseness measures such as the maximal subgraph density or the arboricity of the graph. 
%(see Section~\ref{sec:tech} for a formal definition)
An interesting challenge for dynamic graphs is that the density may vary through the course of dynamic updates, and we prefer not to have the update time in our current sparse graph to be affected by a brief occurrence of density in the past. 
In the work of Henzinger, Neumann, and Wiese, they show how it is possible to adjust to the current graph sparseness in the amortized setting~\cite{henzinger2020explicit}. In this paper, however, we are interested in the case where both the update time is worst-case and the number of re-orientations is bounded. 
One previous approach to this challenge is to take a fixed upper bound on the sparseness as parameter to the algorithm, and then use $\log n$ data structures in parallel~\cite{sawlani2020near}. Since we want the number of re-orientations to also be bounded, we cannot simply change between two possibly very different out-orientations that result from different bounds on the sparseness. Any scheme for deamortising the switch between structures would be less simple than the approach we see in this paper.

\begin{corollary}\label{cor:match}
    There is a deterministic dynamic maximal matching algorithm with worst-case $O(\alpha + \log ^2 n \log \alpha)$ update time, where $\alpha$ is the current arboricity of the dynamic graph. The algorithm also implies a $2$-approximate vertex cover in the same update time. 
\end{corollary}

\subparagraph{Related Work}
Matchings have been widely studied in dynamic graph models. Under various plausible conjectures, we know that a maximum matching cannot be maintained even in the incremental setting and even for low arboricity graphs (such as planar graphs) substantially faster than $\Omega(n)$ update time \cite{AbboudD16,AbboudW14,HenzingerKNS15,KopelowitzPP16,Dahlgaard16}.
Given this, we typically relax the requirement from maximum matching to maintaining matchings with other interesting properties. 
One such relaxation is to require that the maintained matching is only \emph{maximal}. The ability to retain a maximal matching is frequently used by other algorithms, notably it immediately implies a $2$-approximate vertex cover. 
In incremental graphs, maintaining a maximal matching is trivially done with the aforementioned greedy algorithm. 
For decremental\footnote{Maintaining an \emph{approximate maximum matching} decrementally is substantially easier than doing so for fully dynamic graphs. Indeed, recently work by \cite{AssadiBD22} matches the running times for approximate maximum matching in incremental graphs \cite{GLSSS19}. However, for maximal matching, we are unaware of work on decremental graphs that improves over fully dynamic results.} or fully dynamic graphs, there exist a number of trade-offs (depending on whether the algorithm is randomized or determinstic, and whether the update time is worst case or amortized).  Baswana, Gupta, and Sen~\cite{BaswanaGS15} and
Solomon~\cite{Solomon16} gave randomized algorithms maintaining a maximal matching with $O(\log n)$ and $O(1)$ amortized update time. These results were subsequently deamortized by Bernstein, Forster, and Henzinger \cite{BernsteinFH21} with only a $\polylog n$ increase in the update time. For deterministic algorithms, maintaining a maximal matching is substantially more difficult.
Ivkovic and Lloyd~\cite{IvkovicL93} gave a deterministic algorithm with $O((n+m)^{\sqrt{2}/2})$ worst case update time. This was subsequently improved to $O(\sqrt{m})$ worst case update time by Neiman and Solomon \cite{NeimanS16}, which remains the fastest deterministic algorithm for general graphs.

Nevertheless, there exist a number of results improving this result for low-arboricity graphs. Neiman and Solomon \cite{NeimanS16} gave a deterministic algorithm that, assuming that the arboriticty of the graph is always bounded by $\alpha_{\max}$, maintains a maximal matching in amortized time $O(\min_{\beta>1}\{\alpha_{\max} \cdot \beta + \log_{\beta} n\})$, which can be improved to $O(\log n/\log\log n)$ if the arboricity is always upper bounded by a constant. Under the same assumptions, He, Tang, and Zeh \cite{HeTZ14} improved this to $O(\alpha_{\max} + \sqrt{\alpha_{\max}\log n})$ amortized update time.
Without requiring that the arboricity be bounded at all times, the work by Kopelowitz, Krauthgsamer, Porat, and Solomon~\cite{KopelowitzKPS13} implies a deterministic algorithm with $O(\alpha^2 + \log^2 n)$ worst case update time, where $\alpha$ is the arboricity of the graph when receiving an edge-update.

\subsection{Dynamic $\Delta+1$ coloring}

\begin{lemma} \label{thm:folkloreish}
Suppose one can maintain an edge-orientation of a dynamic graph, 
that has $t_u$ update time, 
that for each update performs at most $r_u$ edge re-orientations (direction changes), and that maintains a maximal out-degree of $\le n_o$. Then there is a dynamic $\Delta+1$-coloring algorithm whose update time is $O(t_u+r_u+n_o)$.
\end{lemma}

\begin{proof}
For a vertex $v$, say a color is \emph{in-free} if no in-neighbor of $v$ has that color. 
For a vertex of degree $d$, keep a doubly linked list of in-free colors from the palette 
$0,1,\ldots,d$. Keep an array \texttt{taken} of size $d+1$ where the $i$'th entry points to a doubly-linked list of in-neighbors of color $i$, and an array \texttt{free} of size $d+1$ where the $i$'th entry points to the $i$'th color in the list of in-free colors if the $i$'th color is in-free.

The color of a vertex $v$ is found by finding a color that is both in-free and out-free: examine the $\le n_o$ out-neighbors, and use the \texttt{free}-array to temporarily move the $\le n_o$ out-taken colors to a list \texttt{out-taken}. Give $v$ an arbitrary free color from the remaining list, and undo the \texttt{out-taken} list. This takes $O(n_o)$ time, and gives $v$ a color between $0$ and its degree.

When an edge changes direction, this incurs $O(1)$ changes to linked lists and pointers. When an edge update incurs $r_u$ edge re-orientations, we thus have $O(r_u)$ such changes. When an edge is inserted/deleted from a properly colored graph, at most one vertex needs to be recolored, either because there is a color conflict, or because its color number is larger than its degree. This vertex can be recolored in $O(n_o)$ time. Thus, the total time per edge insertion or deletion is $O(t_u + r_u + n_o)$.
\end{proof}

\begin{corollary}\label{cor:col}
    There is a deterministic dynamic $\Delta+1$ coloring algorithm with worst-case $O(\alpha + \log ^2 n \log \alpha)$ update time, where $\alpha$ is the current arboricity of the dynamic graph. %The algorithm also implies a $2$-approximate vertex cover in the same update time. 
\end{corollary}

\subparagraph{Related Work}
Previous work presented randomized algorithms with constant amortized update time per edge insertion/deletion \cite{BhattacharyaGKL22,HenzingerP22}. For deterministic algorithms, \cite{BhattacharyaCHN18} showed that if one is willing to use $(1+o(1))\cdot \Delta$, colors, a $\polylog(\Delta)$ amortized update time is possible. Solomon and Wein \cite{SolomonW20} extended the algorithm by \cite{BhattacharyaCHN18} and further showed that it is possible to maintain an $\alpha\log^2 n$ coloring in $\polylog(n)$ amortized update time.

%\cite{BhattacharyaGKL22,HenzingerP22,BhattacharyaCHN18,SolomonW20}

\subsection{Dynamic matrix vector product}

Suppose we have an $n \times n$ dynamic matrix $A$, and a dynamic $n$-vector $x$, and we want to maintain a data structure that allows us to efficiently query entries of $Ax$.  
The problem is related to the Online Boolean Matrix-Vector Multiplication (OMV), which is commonly used to obtain conditional lower bounds \cite{CliffordGL15,HenzingerKNS15,LarsenW17,Patrascu10}.
If $A$ is symmetric and sparse, in the sense that the undirected graph $G$ with $A$ as adjacency matrix has low arboricity, then we can use an algorithm for bounded out-degree orientation as a black-box to give an efficient data structure as follows:

%Remark: This statement is not in the published manuscript, but appears in the arxiv version.
\begin{lemma}[Implicit in Thm. A.3 in~\cite{KopelowitzKPS13}]
    Suppose one can maintain an edge-orientation of a dynamic graph with adjacency matrix $A$, that has $t_u$ update time, that for each update performs at most $r_u$ edge re-orientations (direction changes), and that maintains a maximal out-degree of $\le n_o$. Then there is a dynamic matrix-vector product algorithm that supports entry-pair changes to $A$ in $O(t_u+r_u)$ time, entry changes to the vector $x$ in $O(n_o)$ time, and queries to the an entry of product $Ax$ in $O(n_o)$ time.
\end{lemma}
%\todo[inline]{Hmm. Could we generalize this to avoid the symmetry assumption?  I think we could, since that only means we are treating the weight in each direction differently. We just need to define the undirected graph $G$ as having an edge $\{u,v\}$ whenever $A_{uv}\neq0$ or $A_{vu}\neq0$.  The algorithm in the proof already does the right thing, it is the statement of the theorem that gets ugly.}
\begin{proof}
    Let each node $i$ store the sum $s_i=\sum_{j\in N^-(i)}A_{ij}x_j$, i.e. the sum of the terms of $(Ax)_i=\sum_{j\in N(i)}A_{ij}x_j$ corresponding to incoming edges at $i$.  Changing entry $A_{ij}=A_{ji}$ in the matrix to or from $0$ corresponds to deleting or inserting an edge, which takes $t_u$ time and does at most $r_u$ edge re-orientations. Updating the $O(1)$ affected sums after inserting, deleting, re-orienting, or re-weighting an edge takes worst case $O(1)$ time. Any entry update to the matrix $A$ thus takes $O(t_u+r_u)$ time. 
    When a vector entry $x_j$ changes, we need to update the at most $n_o$ sums $\{s_i\}_{i\in N^+(j)}$, which can be done in worst case $O(n_o)$ time.
    Finally, the query for $(Ax)_i$ is computed as $(Ax)_i=s_i+\sum_{j\in N^+(i)}A_{ij}x_j$ in worst case $O(n_o)$ time.
\end{proof}

This result is used in~\cite[Theorem A.3]{KopelowitzKPS13} to give an algorithm for dynamic matrix vector product with running time $O(\alpha^2+\log^2 n)$ for updating the matrix, and $O(\alpha+\log n)$ for updating the vector and for queries.

Combining this theorem with our Theorem~\ref{thm:fast_butworse} gives us an algorithm for dynamic matrix vector product with slightly improved time for updating the matrix:
\begin{corollary}\label{cor:dynamicmatrix}
    Let $A$ be a symmetric $n\times n$ matrix, and let $G$ be the undirected graph whose adjacency matrix is $A$. Let $x$ be an $n$ dimensional vector. Then we can support changes to $A$ in $O(\log^2 n\log\alpha)$ worst case time, changes to $x$ in $O(\alpha+\log n)$ worst case time, and for each $i\in \{1,\ldots,n\}$ we can report $\sum_{j=1}^nA_{ij}x_j$ in worst case $O(\alpha+\log n)$ time.
\end{corollary}

If we instead combine with our Theorem~\ref{thm:fastnonadd} we get an algorithm for dynamic matrix vector product with slightly worse time for updating the matrix, but improved time for updating the vector and for queries:
\begin{corollary}
    Let $A$ be a symmetric $n \times n$ matrix, and let $G$ be the undirected graph whose adjacency matrix is $A$. Let $x$ be an $n$ dimensional vector. Then we can support changes to $A$ in $O(\log^3 n\log\alpha)$ worst case time, changes to $x$ in $O(\alpha)$ worst case time, and for each $i\in \{1,\ldots,n\}$ we can report $\sum_{j=1}^nA_{ij}x_j$ in worst case $O(\alpha)$ time.
\end{corollary}
%For constant $\alpha$ this is strictly better than~\cite{kopelowitz2014orienting}.

\subsection{Dynamic arboricity decomposition}

\begin{lemma}[\cite{henzinger2020explicit,christiansenICALP}] \label{thm:arbdecomp}
Suppose one can maintain an edge-orientation of a dynamic graph, 
that has $t_u$ update time, 
%that for each update performs at most $r_u$ edge re-orientations (direction changes),
and that maintains a maximal out-degree of $\le n_o$. Then there is an algorithm for maintaining a decomposition into $2n_o$ forests whose update time is $O(t_u)$.
\end{lemma}

\begin{proof}
    Firstly, as noted in \cite{henzinger2020explicit,christiansenICALP}: By assigning the $i$'th out-edge of a vertex $u$ to subgraph $S_i$, one obtains a decomposition into $n_o$ subgraphs, each of which is a pseudoforest. Every vertex has at most one out-edge in each pseudoforest $S_i$, and thus, the at most one cycle in each tree of the pseudoforest is a directed cycle according to the orientation.

    For maintaining this dynamic pseudoforest decomposition, there is only an $O(1)$ overhead per edge-reorientation, yielding an $O(n_o)$-time algorithm for maintaining $n_o$ pseudoforests.

    Then, as noted in \cite{henzinger2020explicit}, we may split each pseudoforest $S_i$ into two forests $f_i$ and $f_i'$ by the following simple algorithm: given a new edge $e$ in $f_i$, notice that there is at most one edge $e'$ in $f_i$ incident to its head. Now, one can safely insert $e$ in any of the two forests $\{f_i , f_i'\}$ that does not contain this at most one edge $e'$. Thus, consequently, neither $f_i$ nor $f_i'$ will contain a cycle.
\end{proof}

Thus, by applying Theorem~\ref{thm:fastnonadd}, we obtain the following:

\begin{corollary}\label{cor:decomp}
    There is a deterministic algorithm for maintaining an arboricity decomposition into $O(\alpha)$ forests, whose worst-case update time is $O(\log ^3 n \log \alpha)$,
    where $\alpha$ is the current arboricity of the dynamic graph.
\end{corollary}

\subparagraph{Related Work}
While an arboricity decomposition of a graph; a partition of its edges into minimally few forests; is conceptually easy to understand, computing an arboricity decomposition is surprisingly nontrivial. Even computing it exactly has received much attention~\cite{GabowW,Gabow95,Edmonds1965MinimumPO,PicardQueyranne82}. 
The state-of-the-art for computing an exact arboricity decomposition runs in $\tilde{O}(m^{3/2})$ time \cite{GabowW,Gabow95}.
In terms of not-exact algorithms there is a 2-approximation algorithm~\cite{ArikatiMZ97,Eppstein94} as well as an algorithm for computing an $\alpha+2$ arboricity decomposition in near-linear time~\cite{blumenstock2019constructive}.

For dynamic arboricity decomposition, Bannerjee et al.~\cite{Banerjee} give a dynamic algorithm for maintaining the current arboricity. The algorithm has a near-linear update time. They also provide a lower bound of $\Omega(\log{n})$. % for dynamically maintaining arboricity.
Henzinger Neumann Wiese~\cite{henzinger2020explicit} provide an $O(\alpha)$ arboricity decomposition in $O(\poly(\log n , \alpha))$ time; their result also goes via out-orientation, and they provide a dynamic algorithm for maintaining a $2\alpha'$ arboricity decomposition, given access to any black box dynamic $\alpha'$ out-degree orientation algorithm. 
Most recently, there are algorithms for maintaining $(\alpha+2)$ forests in $O(\operatorname{poly} (\log(n) , \alpha))$ update-time~\cite{christiansenICALP}, and $(\alpha + 1)$ forests in $\tilde{O}(n^{3/4}\operatorname{poly}(\alpha))$ time~\cite{christiansenMFCS}.

\bibliographystyle{alpha}
\bibliography{refs}

\end{document}

%% file: introtable.tex
\begin{table}[h]
\vspace{1em}
    \centering
    \begin{tabularx}{1.1\textwidth}{c|c|c|c|l}
      \makecell{\textbf{Dynamic} \\
     \textbf{Problem} 
     } & \textbf{Guarantee} & \textbf{Worst-case Time} & \textbf{Thm.} & \textbf{State-of-the-art comparison} \\
         \hline
         \hline
     %
     %   CONSTANT FACTOR
     %
    \makecell{out-orient \\ /density}
     & 
     \makecell{  $O(\alpha)$ \\  $O(\rho)$}      & $O( \log^3 n \log \alpha)$ & Thm.~\ref{thm:fastnonadd} & 
     \makecell[l]{
       $O(\alpha_{\mx}) \textnormal{ in } O(\alpha_{\mx} + \log n) \sim \textnormal{amor.~\cite{Brodal99dynamicrepresentations}}$ \\
      $O(\alpha) \textnormal{ in } O(\log^2 n) \sim \textnormal{amor.~\cite{henzinger2020explicit}}$ \\
      $O(\alpha) \textnormal{ in } O(\log^4 n) \sim \textnormal{implicit~\cite{sawlani2020near}}$
     }\\
     \hline
     % 
     % ADDITIVE LOG
     %
     \makecell{out-orient \\ /density} & 
      \makecell{  $O(\alpha + \log n)$ \\  $O(\rho + \log n)$} & $O(\log ^2 n \log \alpha)$ & Thm.~\ref{thm:fast_butworse} &      
     \makecell[l]{
       $O(\alpha_{\mx} + \log n) \textnormal{ in } O(\log n)  \textnormal{~\cite{berglinetal:LIPIcs:2017:8263}}$ \\
      $O(\alpha + \log n) \textnormal{ in } O(  \alpha^2 + \log^2 n ) \textnormal{~\cite{KopelowitzKPS13}}$\\
      $O(\alpha) \textnormal{ in } O(\log^4 n) \sim \textnormal{implicit~\cite{sawlani2020near}}$
} \\
     \hline
     %
     % EPSILON DENSITY
     %
     density & $ (1 + \eps) \rho $ & $O\left(\varepsilon^{-6}\log^3 n \log \rho \right)$ & Cor.~\ref{cor:eps_subgraphdensity} & 
     \makecell[l]{
      $(1 + \eps) \rho \textnormal{ in } O(\eps^{-6} \log^4 n  ) \textnormal{~\cite{sawlani2020near}} $
     } \\
     \hline
    %
     % 2 EPSILON OUT 
     %
    out-orient & $(2 + \eps) \alpha + 1$ & $O\left(\varepsilon^{-6}\log^3 n \log \alpha \right)$ & Obs.~\ref{obs:rounding} & 
        $(2 + \eps) \alpha \textnormal{ in } O(\eps^{-6} \log^4 n  ) \sim \textnormal{implicit~\cite{sawlani2020near}}$
     \\
     %%%%%
     % 2 EPS OUT
     %%%%
     \hline
     out-orient & $(1 + \eps)\alpha + 2$ & $O( \eps^{-6} \log^3 n \log \alpha)$ & Thm.~\ref{thm:epsapprox} & 
     $
     O( \eps^{-6} \alpha^2 \log^3 n) \textnormal{~\cite{christiansenICALP}}
     $ \\
     %
     % MAX MATCHING
     %
     \hline
      \makecell{maximal \\ matching}  & -- & $O(\alpha + \log ^2 n \log \alpha)$ & Cor.~\ref{cor:match} & 
     \makecell[l]{$%
    O(\alpha_{\mx} + \log n) \textnormal{~ \cite{berglinetal:LIPIcs:2017:8263}}$ \\
       $O(\alpha_{\mx}+\sqrt{\alpha_{\mx}\log n} ) \sim\textnormal{amor.~\cite{HeTZ14}}$\\
       $O(\sqrt{m} )  \textnormal{~\cite{NeimanS16}}$ \\
     $O(\alpha^2 + \log^2 n) \textnormal{~\cite{KopelowitzKPS13}}%
     $} \\
     % 
     % DYNAMIC MATRIX MULTI
     %
       \hline
     \makecell{ maintain \\
     $A \cdot \overrightarrow{x}$} & -- & \makecell{$O(\log ^2 n \log \alpha)$ 
     for $A$ 
     \\
     $O(\alpha + \log n)$  \textnormal{ for }  $\overrightarrow{x}$ \\
     $O(\alpha + \log n)$  \textnormal{ query } 
     } & Cor.~\ref{cor:dynamicmatrix} &
%     \makecell[l]{
    $\begin{cases}
     O(\alpha^2 + \log^2 n) \textnormal{ updating } A \\
     O(\alpha + \log n) \textnormal{ query+upd } \overrightarrow{x} 
     \end{cases}$
     \cite{KopelowitzKPS13}
%     }
      \\
     %
     % FOREST DECOMP
     %
     \hline
     \makecell{ arboricity \\
     decomp } 
     & $O(\alpha)$ & $O(\log^3 n \log \alpha)$ & Cor.~\ref{cor:decomp} &
     \makecell[l]{
     $O(\alpha) \textnormal{ in } O(\log^2 n) \sim \textnormal{amor. }$
     \cite{henzinger2020explicit} \\ 
     $O(\sqrt{m}\log n)$ \cite{Banerjee}
     }
     \\
     \hline
    \makecell{ $\Delta+1$ \\
     coloring } 
     &  & $O(\alpha+\log^2 n\log \alpha)$ & Cor.~\ref{cor:col} &
     $O(\Delta)$ (folklore)
    \end{tabularx}
    \caption{
    Our results, compared to the state-of-the art. Unless mentioned otherwise, all algorithms are deterministic, have a worst case running times, and are fully-adaptive to the aboricity, and all orientations are stored explicitly.
    We denote by $\alpha$ the arboricity of the graph and by $\rho$ the maximal subgraph density. 
%    By using the notation $\widetilde{O}(f(n))$, we hide factors polynomial in $\log f(n)$.
    In the column comparing with state-of-the-art, we have indicated with ``$\sim$ implicit'' when the cited paper only gives implicit access to the solution it is maintaining, and we have indicated with ``$\sim$ amor.'' when the cited paper only gives an amortized running time. 
    \label{tab:results}
    }
\end{table}